\documentclass[lettersize,onecolumn]{IEEEtran}
\usepackage{amsmath,amsfonts}
\usepackage{algorithmic}
\usepackage{algorithm}
\usepackage{array}
\usepackage[caption=false,font=normalsize,labelfont=sf,textfont=sf]{subfig}
\usepackage{textcomp}
\usepackage{stfloats}
\usepackage{url}
\usepackage{verbatim}
\usepackage{graphicx}
\usepackage{cite}
\usepackage{amsthm}         
\usepackage{float}
\usepackage{newtxtext}
\usepackage{graphicx}       
\usepackage{amsmath}        
\usepackage{amssymb,bm,cleveref}        
\usepackage{amsfonts} 
\usepackage{cancel}

\DeclareMathOperator*{\argmax}{arg\,max}
\DeclareMathOperator*{\argmin}{arg\,min}

\DeclareMathOperator{\dssnr}{\delta_{\text{S}}}

\DeclareMathOperator{\tr}{tr}

\DeclareMathOperator{\im}{Im}

\DeclareMathOperator{\Sp}{Sp}

\newtheorem{theorem}{Theorem}

\newtheorem{lemma}{Lemma}
\newtheorem{remark}{Remark}

\newtheorem{proposition}{Proposition}

\def\by{{\bm{y}}}

\def\R{{\mathbb{R}}}
\def\N{{\mathbb{N}}}

\def\nats{{\mathbb{N}}}
\def\bA{{\bm{A}}}
\def\bB{{\bm{B}}}

\def\bw{{\bm{w}}}

\def\bU{{\bm{U}}}
\def\bx{\bm{x}}
\def\bz{\bm{z}}

\def\b0{{\bm{0}}}
\def\bX{{\bm{X}}}

\def\bI{{\bm{I}}}

\def\bP{{\bm{P}}}

\def\bT{{\bm{T}}}
\def\bSigma{{\bm{\Sigma}}}

\def\calI{{\mathcal{I}}}

\def\calX{{\mathcal{X}}}
\def\calK{{\mathcal{K}}}
\def\psdgeq{\succeq}
\def\psdleq{\preceq}

\def\eps{\epsilon}

\begin{document}

\title{The Sharp Phase Transition of Tyler's M-Estimator for Robust Subspace Recovery}

\author{
    Gilad Lerman\thanks{School of Mathematics, University of Minnesota, Minneapolis, MN 55455 USA (e-mail: lerman@umn.edu).} and 
    Teng Zhang\thanks{Department of Mathematics, University of Central Florida, Orlando, FL 32816 USA (e-mail: teng.zhang@ucf.edu).}
}



\maketitle

\begin{abstract}
Robust Subspace Recovery (RSR) aims to identify an underlying $d$-dimensional subspace from a dataset heavily corrupted by outliers. Complexity-theoretic results establish a threshold for the problem's computational hardness based on the dimension-scaled signal-to-noise ratio (DS-SNR): 
the problem is SSE-hard when the DS-SNR is strictly less than 1, and solvable via practical algorithms when it is greater than 1 under general position assumptions. However, the exact behavior of practical algorithms at the critical boundary DS-SNR = 1 has remained unknown. This work resolves the behavior of Tyler's M-estimator (TME) at this critical boundary, consequently establishing a sharp phase transition. Specifically, we prove that TME converges exactly to the true subspace for DS-SNR $\geq$ 1 under a new stability condition, which is less restrictive than the general position assumptions used in prior literature. Our analysis utilizes a decomposition of the TME iterates within a majorization-minimization framework.
\end{abstract}


\section{Introduction}

Robust Subspace Recovery (RSR) is a fundamental problem in robust statistics, machine learning, and computer vision. The primary goal of RSR is to identify an underlying low-dimensional linear subspace from a dataset that is heavily corrupted by outliers. The standard formulation of the noiseless RSR problem assumes a dataset $\calX = \{\bx_i\}_{i=1}^N \subset \R^D$ consisting of $n_1$ inliers lying exactly on a $d$-dimensional linear subspace $L_* \subset \R^D$, and $n_0$ outliers lying strictly off $L_*$. We refer to such a dataset as a noiseless inlier-outlier dataset, where the total number of points is $N=n_0+n_1$. The central algorithmic question in noiseless RSR is under what conditions one can exactly and efficiently recover the underlying $d$-subspace $L_*$.

A natural metric for characterizing the difficulty of this problem is the ratio of inliers to outliers, $n_1/n_0$, which can be viewed as a signal-to-noise ratio (SNR)~\cite{lerman2018overview,maunu19ggd,maunu2019robust}. Building on foundational results~\cite{hardt2013algorithms,zhang2016robust}, Yu et al.~\cite{Yu2024} suggested scaling $n_1$ and $n_0$ by the dimension $d$ and codimension $D-d$ of the inlier subspace, respectively. This leads to the \textit{dimension-scaled SNR (DS-SNR)}, denoted by $\dssnr$:
\begin{equation}
\label{eq:dssnr}
    \dssnr := \frac{n_1/d}{n_0/(D-d)}.
\end{equation}

The DS-SNR precisely characterizes the computational hardness of the RSR problem. Hardt and Moitra~\cite{hardt2013algorithms} established a fundamental lower bound, showing that when $\dssnr < 1$, the noiseless RSR problem is Small Set Expansion (SSE)-hard, a property conjectured to be equivalent to NP-hardness~\cite{raghavendra2010graph}. In the special case of hyperplanes ($d = D-1$), they showed NP-hardness by invoking a result from \cite{Khachiyan_complexity1995}.

\begin{theorem}[\cite{hardt2013algorithms}]
\label{thm:sse_hard}
The noiseless RSR problem is SSE-hard if $\dssnr < 1$. Furthermore, if $\dssnr < 1$ and $d = D-1$, then the problem is NP-hard.
\end{theorem}

Conversely, when $\dssnr > 1$, the problem becomes computationally tractable under certain strong structural assumptions. Hardt and Moitra~\cite{hardt2013algorithms} proposed probabilistic and deterministic subspace-search algorithms (RandomizedFind and DeRandomizedFind, respectively), which recover $L_*$ provided the dataset is in \textit{general position} with respect to $L_*$. This means that every subset of $D$ points is linearly independent if and only if it contains at most $d$ inliers, effectively assuming $L_*$ is the only low-dimensional structure in the data. Since the focus of this paper is on deterministic algorithms, we formulate the theorem for DeRandomizedFind. We introduce two necessary conditions that are absent from \cite{hardt2013algorithms} and clarified later in \Cref{rem:hardt_moitra_conditions} of \Cref{sec:clarify_zhang16_extend}. 

\begin{theorem}[\cite{hardt2013algorithms}, Theorem 22]
\label{thm:hm_doable}
Assume $\calX \subset \R^D$ is a noiseless inlier-outlier dataset with an underlying $d$-subspace $L_*$. If $\calX$ is in general position with respect to $L_*$, $\dssnr > 1$, $N>D$, and there are at least $D-d$ outliers, DeRandomizedFind outputs $L_*$ in polynomial time.
\end{theorem}
Zhang~\cite{zhang2016robust} later established that Tyler's M-Estimator (TME)~\cite{tyler1987distribution}, a widely used robust statistical technique, can also exactly recover $L_*$ under slightly different general position assumptions. Specifically, we say that the inliers are in general position with respect to $L_*$ if 
any $d$ of them are linearly independent. Similarly, we say that the outliers are in
general position with respect to $L_*^\perp$ if after projecting them onto $L_*$ any $D-d$ of them are linearly independent.
\begin{theorem}[\cite{zhang2016robust}]
\label{thm:zhang16}
Let $\calX$ be a noiseless inlier–outlier dataset not containing the origin, with an underlying $d$-subspace $L_*$. Assume the inliers are in general position with respect to $L_*$ and the outliers are in general position with respect to $L_*^\perp$. If $\dssnr > 1$, then TME exactly recovers $L_*$.
\end{theorem}

While the feasibility of RSR is well-understood for $\dssnr > 1$ and $\dssnr < 1$, the exact algorithmic behavior at the critical boundary $\dssnr = 1$ has remained a significant open problem. Furthermore, the strong general position assumptions utilized in prior works limit the applicability of these guarantees in practical scenarios where data often exhibits unmodeled local collinearities.

\subsection{Tyler's M-Estimator (TME)}
\label{sec:TME_first_review}

For a  dataset $\calX = \{\bx_i\}_{i=1}^N \subset \R^D$, TME produces an empirical estimate of a robust autocorrelation matrix, determined up to a positive scalar multiple (i.e., a robust shape matrix). We remark that the empirical autocorrelation matrix (up to scale) is $\sum_{\bx \in \calX} \bx \bx^\top$. Our discussion does not require any centering of the given dataset and thus this matrix may differ from the empirical covariance matrix.  
The TME estimator is defined as a minimizer of the non-convex energy function:
\begin{equation}\label{eq:TME_objective}
F(\bSigma) = \frac{1}{N} \sum_{\bx \in \calX} \log(\bx^\top \bSigma^{-1} \bx) + \frac{1}{D} \log\det(\bSigma)
\end{equation}
over the set of positive definite matrices $S_{++}(D)$. Because $F(c\bSigma) = F(\bSigma)$ for any $c > 0$, the minimizer is only identifiable up to scale. Furthermore, the minimizer and the resulting iterative updates are invariant to the scaling of individual data points $\bx_i \mapsto c_i \bx_i$ (for $c_i \neq 0$). We note that $F(\bSigma)$ is undefined when $\bx = \b0$ due to the $\log(\bx^\top \bSigma^{-1} \bx)$ term. Thus, we assume throughout the paper that $\calX$ does not contain the origin. This is not a restriction: if $\b0$ is included in the dataset, one can simply estimate the robust autocorrelation matrix using the non-zero points, and use these to trivially obtain the estimates for the entire dataset.

The standard TME algorithm \cite{tyler1987distribution} aims to minimize this objective by starting with an arbitrary $\bSigma^{(0)} \in S_{++}(D)$ and iteratively computing for $k \geq 0$: 
\begin{equation}
\label{eq:tme_iteration}
\bSigma^{(k+1)} =T(\bSigma^{(k)}):= 
\frac{\sum_{i=1}^N\frac{\bx_i\bx_i^\top}{\bx_i^\top\left(\bSigma^{(k)}\right)^{-1}\bx_i}}{\tr\left(\sum_{i=1}^N\frac{\bx_i\bx_i^\top}{\bx_i^\top\left(\bSigma^{(k)}\right)^{-1}\bx_i}\right)}\,.
\end{equation}
If the dataset $\calX$ does not span $\R^D$, $\bSigma^{(k)}$ operates strictly within $\operatorname{span}(\calX)$ utilizing the pseudoinverse (though, as shown later in Lemma \ref{lem:implicit_dimension}, our theoretical conditions guarantee $\calX$ fully spans $\R^D$ when $\dssnr=1$). While the TME estimator is defined up to scale, we find it convenient to use the trace-normalized formulation above, ensuring the estimator remains uniformly bounded in terms of spectral or Frobenius norms across all iterations. Under certain assumptions, the existence and uniqueness of TME were established in \cite{tyler1987distribution}  and \cite[Proposition 1(a)]{tyler2005}), and the convergence of TME was established in \cite[Theorem 2]{Kent_Tyler88}.

For the inlier set $\calX_{\mathrm{in}} \subset L_*$, let $\bU_{L_*}$ be a $D \times d$ matrix with orthonormal columns spanning $L_*$. While this matrix is not unique, we fix one such matrix throughout the paper. The dimension-reduced inlier set in $\R^d$ is defined as $\tilde{\calX}_{\mathrm{in}} := \{\bU_{L_*}^\top\bx: \bx\in \calX_{\mathrm{in}}\}$. We denote the trace-normalized TME estimator for $\tilde{\calX}_{\mathrm{in}}$ (uniquely determined by enforcing $\tr(\bSigma_{\mathrm{in},*}) = 1$) by $\bSigma_{\mathrm{in},*} \in S_{++}(d)$. \Cref{lem:existence_uniqueness} in the Appendix formally establishes that the stability condition introduced in our forthcoming main theorem (Theorem~\ref{thm:zhang16_extend}) guarantees this normalized estimator exists and is unique.

\subsection{Main Results and Technical Contributions}\label{sec:main}

The primary contribution of this paper is establishing the strongest known exact recovery conditions for efficient RSR to date. Most notably, we resolve the open question regarding the critical boundary by proving that TME exhibits a sharp phase transition exactly at $\dssnr = 1$. Our main theorem is formulated as follows:

\begin{theorem}
\label{thm:zhang16_extend}
Assume a noiseless inlier-outlier dataset $\calX \subset \R^D$ that does not contain the origin, with an underlying $d$-subspace $L_*$. Assume further that TME is arbitrarily initialized with $\bSigma^{(0)}\in S_{++}(D)$. Suppose the following conditions hold: 
\begin{enumerate}
    \item $\dssnr \geq 1$. \label{item:dssnr1}
    \item
   \( L_* \) is the unique subspace that achieves the largest \( \dssnr \). That is, for every subspace \( L \subsetneq \R^D \) with \( L \neq L_* \) and $\dim(L) \geq 1$, 
\begin{equation}\label{eq:assumption3_thm}
\frac{|\mathcal{X} \cap L|}{\dim(L)} < \frac{|\mathcal{X} \cap L_*|}{\dim(L_*)}.
\end{equation}
\label{item:convergence}
\end{enumerate}
Then the TME sequence converges exactly to the true subspace. Specifically, $\lim_{k\to\infty}\bSigma^{(k)}=\bU_{L_*}\bSigma_{\mathrm{in},*}\bU_{L_*}^\top$, where $\bSigma_{\mathrm{in},*}\in\R^{d\times d}$ is the TME solution for the projected inliers $\tilde{\calX}_{\mathrm{in}}$ discussed in \Cref{sec:TME_first_review}.
\end{theorem}
\Cref{thm:zhang16_extend} explicitly guarantees exact algorithmic recovery for the previously uncharacterized case of $\dssnr=1$. It also improves upon the recovery conditions in Theorem~\ref{thm:zhang16} for $\dssnr > 1$ by replacing restrictive general position assumptions with \cref{item:convergence}. We demonstrate in the supplementary material (via Lemma \ref{lem:implicit_dimension}) that this assumption  implicitly guarantees that $N > D$ and, in the critical case of $\dssnr=1$, the data fully spans $\R^D$. Following the terminology of \cite{lerman15reaper,maunu19ggd}, this is a stability condition ensuring the permeance of inliers while restricting the alignment of outliers.

As demonstrated by the following  Propositions~\ref{prop:moitra} and \ref{prop:thirdassumption} (proved in \cref{sec:clarify_zhang16_extend}), \cref{item:convergence} is significantly less restrictive than the constraints required by Theorems~\ref{thm:hm_doable} and~\ref{thm:zhang16}. In fact, the two assumptions of \Cref{thm:zhang16_extend} can be viewed as the weakest possible assumptions. If the first one is violated (i.e., $\dssnr < 1$), then by Theorem~\ref{thm:sse_hard},  efficient exact recovery is impossible. If the second one does not hold, there exists another subspace distinct from $L_*$ with the same $\dssnr$ value; consequently, the two subspaces cannot be distinguished, making unique recovery of $L_*$ impossible.

\begin{proposition}\label{prop:moitra}
Suppose $N > D$, $\dssnr\geq 1$, there are at least $D-d$ outliers, and the stability assumption of Theorem~\ref{thm:hm_doable} holds ($\calX$ is in general position with respect to $L_*$). Then \cref{item:convergence} in \Cref{thm:zhang16_extend} holds.
\end{proposition}

\begin{proposition}\label{prop:thirdassumption}
If $\dssnr \geq 1$ (with \(N > D\) whenever $\dssnr=1$), and the stability assumptions of \Cref{thm:zhang16} hold (the inliers are in general position relative to \(L_*\) and the outliers relative to \(L_*^\perp\)), then \cref{item:convergence} in \Cref{thm:zhang16_extend} naturally holds.
\end{proposition}

Beyond improving statistical recovery thresholds, this paper introduces several technical innovations to the analysis of TME:
\begin{itemize}
    \item \textbf{Equivalence of MM Algorithms:} We introduce an alternative Majorization-Minimization (MM) formulation for TME and prove that its iterative steps are exactly equivalent to the standard TME iterations in \eqref{eq:tme_iteration}. While equivalence at the objective function level has been observed previously~\cite{wiesel2011unified,Zhang2014}, establishing equivalence precisely at the level of the iterative algorithm is novel and facilitates a highly refined convergence analysis.
    \item \textbf{Novel Objective Decomposition:} To successfully handle the delicate boundary case at $\dssnr = 1$, we introduce a novel decomposition of the TME objective function and iterations in Lemma~\ref{lemma:step2}, allowing us to isolate and bound the asymptotic behavior of inliers versus outliers.
\end{itemize}

\subsection{Notation}
\label{sec:notation}

Throughout the paper, $\calX = \{\bx_i\}_{i=1}^N \subset \R^D$ denotes the RSR dataset, and $L_*$ denotes its underlying $d$-dimensional subspace. We denote the set of inliers by $\calX_{\mathrm{in}} = \calX \cap L_*$ ($n_1 = |\calX_{\mathrm{in}}|$) and the set of outliers by $\calX_{\mathrm{out}} = \calX \setminus \calX_{\mathrm{in}}$ ($n_0 = |\calX_{\mathrm{out}}|$). 

We use $C$ and $c$ to denote generic positive constants, which may depend on $\calX$ and the initialization $\bSigma^{(0)}$. Generally, $C$ represents a "large" constant utilized for upper bounds, while $c$ represents a "small" constant utilized for lower bounds. 

Bold uppercase and lowercase letters denote matrices and column vectors, respectively. For a matrix $\bA$, $\tr(\bA)$ denotes its trace, $\im(\bA)$ its image, and $\|\bA\|$ its spectral norm. $S_{+}(D)$ and $S_{++}(D)$ denote the sets of positive semidefinite and positive definite matrices in $\R^{D \times D}$, respectively. The eigenvalues of $\bA \in S_{+}(D)$ are ordered as $\sigma_1(\bA) \geq \sigma_2(\bA) \geq \cdots \geq \sigma_D(\bA)$. $\bI_k$ denotes the $k \times k$ identity matrix. For $k \leq D$, $O(D,k)$ denotes the set of semi-orthogonal matrices ($\bU \in \R^{D \times k}$, $\bU^\top \bU = \bI_k$). $O(k)$ denotes the orthogonal matrices in $\R^{k \times k}$.

For a $d$-dimensional subspace $L \subset \R^D$, let $\bP_L$ denote the orthogonal projection matrix onto $L$. We associate with $L$ a matrix $\bU_L \in O(D,d)$ having orthonormal columns such that $\bU_L \bU_L^\top = \bP_L$. While the projection matrix $\bP_L$ is uniquely determined by the subspace $L$, the basis matrix $\bU_L$ is only unique up to right-multiplication by a $d \times d$ orthogonal matrix. To maintain consistency, we assume that for each subspace $L$, an arbitrary choice of $\bU_L$ is made and fixed throughout the paper.

Finally, for any $d_1$-dimensional subspace $L_1$ and $d_2$-dimensional subspace $L_2$, we define the block matrix component mapping:
\begin{equation}
\label{eq:def_block_components}
[\bSigma]_{L_1,L_2} \;=\; \bU_{L_1}^\top \bSigma \bU_{L_2}\in\R^{d_1\times d_2}. 
\end{equation}

\subsection{Paper Organization}
\label{sec:organization}
The remainder of the paper is organized as follows. \Cref{sec:MM_perspective} introduces the Majorization-Minimization framework, establishes its equivalence with TME, and proves a property of the limit points of the TME iterates.  
\Cref{sec:proof_thm_zhang16_extend,sec:proof_thm_zhang16_extend_part2} provide the proof of our primary theoretical guarantee (\Cref{thm:zhang16_extend}) when $\dssnr=1$ and $\dssnr>1$, respectively. Finally, \Cref{sec:conclusion} concludes the paper. \Cref{sec:supplement} provides additional clarifications and proofs of different stated claims.

\section{A Majorization-Minimization Perspective of TME}\label{sec:MM_perspective}

In this section, we show that the traditional TME sequence can be derived and analyzed through a specialized majorization-minimization (MM) framework. We first provide essential background on MM in the context of the TME objective function. We then present an alternative formulation for TME following \cite{wiesel2011unified,Zhang2014}, and crucially, we prove that the iterative algorithmic steps of this alternative framework are exactly equivalent to the standard fixed-point algorithm defined in \eqref{eq:tme_iteration}. 

We note that this strict algorithmic equivalence has not been proven previously. Wiesel~\cite{wiesel2011unified} introduced the objective function of this new formulation through intuitive motivations but did not discuss its optimization. Zhang et al.~\cite{Zhang2014} showed that the TME estimator is a weighted sum of $\bx\bx^\top$ with weights obtained from the minimizer of this alternative objective, but did not establish the equivalence of the iterative algorithms themselves.

Recall that the standard TME algorithm aims to minimize $F(\bSigma)$ by iteratively applying the update rule in \eqref{eq:tme_iteration}. We now demonstrate why this procedure naturally forms an MM sequence for $F(\bSigma)$, following the approaches of \cite[Section~II]{wiesel2011unified} and \cite[Section~3.1]{zhang2016robust}. 

For $\bSigma_0, \bSigma \in S_{++}(D)$, define the surrogate function:
\[
G(\bSigma,\bSigma_0) = \frac{1}{N} \sum_{\bx \in \calX} \frac{\bx^\top\bSigma^{-1}\bx}{\bx^\top \bSigma_0^{-1} \bx} + \frac{1}{N} \sum_{\bx \in \calX} \log(\bx^\top \bSigma_0^{-1} \bx) + \frac{1}{D} \log\det(\bSigma) - 1.
\]
By applying the standard logarithmic inequality,
\[
\log x \le \log a + \frac{x}{a} - 1, \quad \text{for } x,a > 0,
\]
which holds with equality if and only if $x=a$, we see that $G$ majorizes $F$ in the following sense: 
\[
F(\bSigma) \leq G(\bSigma,\bSigma_0) \quad \text{for all } \bSigma_0, \bSigma \in S_{++}(D),
\]
with equality holding if and only if $\bSigma = \bSigma_0$. 

The unnormalized minimizer of the surrogate function $G$ is given by:
\[
\widetilde{\bSigma}^{(k+1)} = \argmin_{\bSigma\in S_{++}(D)} G(\bSigma,\bSigma^{(k)}).
\]
As indicated by the derivative in \eqref{eq:derivative_G}, this unnormalized minimizer evaluates to a matrix with a trace of $D$. However, because the objective function $F$ is strictly scale-invariant ($F(\alpha \bSigma) = F(\bSigma)$ for any $\alpha > 0$), we can project this minimizer back to the unit-trace set without altering the objective value. The standard TME iterative update \eqref{eq:tme_iteration} can thus be cleanly expressed as the normalized minimizer:
\[
\bSigma^{(k+1)} = \frac{\widetilde{\bSigma}^{(k+1)}}{\tr\big(\widetilde{\bSigma}^{(k+1)}\big)}.
\]
The validity of this update equation follows from the expression for the partial derivative of $G$, provided later in \eqref{eq:derivative_G}, and the fact that for a fixed $\bSigma^{(k)}$, this derivative has a single zero. 

From this majorization-minimization perspective, the sequence of objective values $\{F(\bSigma^{(k)})\}_{k\geq 1}$ is monotonically nonincreasing:
\begin{equation}\label{eq:monotone}
F(\bSigma^{(k+1)}) \leq G(\bSigma^{(k+1)},\bSigma^{(k)}) \leq G(\bSigma^{(k)},\bSigma^{(k)}) = F(\bSigma^{(k)}).
\end{equation}

An alternative formulation of TME, derived from the profile likelihood of the deterministic scaled Gaussian model \cite{wiesel2012geodesic}, focuses on the vector of weights $\mathbf{w}=(w_1,\ldots,w_N)$ instead of the empirical autocorrelation matrix $\boldsymbol{\Sigma} = \sum_{i=1}^N w_i \mathbf{x}_i\mathbf{x}_i^\top$. Specifically, it considers the minimization of 
$\tilde{F}: \mathbb{R}_{++}^N \rightarrow \mathbb{R}$, where $\mathbb{R}_{++}^N = \{ \mathbf{w} \in \mathbb{R}^N \mid w_i > 0 \text{ for all } i \}$, defined by
\begin{equation}\label{eq:TME_objective2}
\tilde{F}(\bw) = -\frac{1}{N}\sum_{i=1}^N \log w_i + \frac{1}{D}\log\det\left(\sum_{i=1}^N w_i\bx_i\bx_i^\top\right).
\end{equation}
Similar to $F$ in \eqref{eq:TME_objective}, $\tilde{F}$ is scale-invariant; that is, for any constant $\alpha>0$, $\tilde{F}(\alpha\bw) = \tilde{F}(\bw)$. The minimizer of $\tilde{F}$, denoted $\bw^*=(w^*_1,\dots,w^*_N)$, yields the TME estimator (the minimizer $\bSigma^*$ of $F$) via $\bSigma^* = \sum_{i=1}^N w_i^*\bx_i\bx_i^\top$. To verify this relationship, note that the minimizer $\bw^*$ of the log-convex function $\tilde{F}$ satisfies $\nabla_{\bw}\tilde{F}(\bw^*) = 0$, or equivalently,
\[
\frac{1}{N w_i^*} = \frac{1}{D} \, \bx_i^\top\left(\sum_{j=1}^N w_j^*\bx_j\bx_j^\top\right)^{-1} \bx_i, \qquad \text{for all } 1\leq i\leq N.
\]
Setting $\tilde{\bSigma} = \sum_{i=1}^N w_i^*\bx_i\bx_i^\top$, we find $w_i^* = \frac{D}{N} \frac{1}{\bx_i^\top\tilde{\bSigma}^{-1}\bx_i}$, which leads to
\begin{equation}\label{eq:fixed_from_Ftilde}
\tilde{\bSigma} = \sum_{i=1}^N w_i^*\bx_i\bx_i^\top = \frac{D}{N} \sum_{i=1}^N \frac{\bx_i\bx_i^\top}{\bx_i^\top\tilde{\bSigma}^{-1}\bx_i}.
\end{equation}
This is equivalent to $\nabla_\bSigma F(\tilde{\bSigma}) = 0$. Combining this observation with the geodesic convexity of $F$ \cite{wiesel2011unified,zhang2016robust} implies that $\tilde{\bSigma} = \bSigma^*$.

Next, we formulate an iterative update for $\bw$ using the MM principle with the majorization function
\[
\tilde{G}(\bw,\bw') = -\frac{1}{N}\sum_{i=1}^N \log w_i + \frac{1}{D}\tr\left(\left(\sum_{j=1}^N w_j'\bx_j\bx_j^\top\right)^{-1}\sum_{i=1}^N w_i\bx_i\bx_i^\top\right) + C,
\]
where $C = \frac{1}{D}\log\det\left(\sum_{j=1}^N w_j'\bx_j\bx_j^\top\right) - 1$ is chosen so that $\tilde{G}(\bw',\bw') = \tilde{F}(\bw')$. 
We verify that $\tilde{G}$ majorizes $\tilde{F}$, meaning $\tilde{G}(\bw,\bw') \geq \tilde{F}(\bw)$ with equality when $\bw=\bw'$. Recognizing that $\log\det(\bX)$ is concave, its first-order upper bound at $\bX=\bX_0$ is given by $\log\det(\bX_0) + \tr(\bX_0^{-1}(\bX-\bX_0))$. Applying this concavity inequality to the right-hand side (RHS) of \eqref{eq:TME_objective2} immediately yields $\tilde{G}(\bw,\bw') \geq \tilde{F}(\bw)$. We thus obtain the MM update formula:
\begin{equation*}
{\bw}^{(k+1)} = \argmin_{\bw} \tilde{G}(\bw,\bw^{(k)}),
\end{equation*}
which evaluates explicitly to
\begin{equation}\label{eq:w_update}
w_i^{(k+1)} = \frac{D}{N} \left( \mathbf{x}_i^\top \left( \sum_{j=1}^N w_j^{(k)} \mathbf{x}_j\mathbf{x}_j^\top \right)^{-1} \mathbf{x}_i \right)^{-1}.
\end{equation}
Just as in \eqref{eq:monotone}, this guarantees that the sequence $\tilde{F}(\bw^{(k)})$ is nonincreasing:
\[
\tilde{F}(\bw^{(k+1)}) \leq \tilde{G}({\bw}^{(k+1)},\bw^{(k)}) \leq \tilde{G}(\bw^{(k)},\bw^{(k)}) = \tilde{F}(\bw^{(k)}).
\]

We further note that the update for $\bw$ in \eqref{eq:w_update} is perfectly equivalent to the update for $\bSigma$ in \eqref{eq:tme_iteration}, up to a scaling factor. Indeed, if we define
\begin{equation}\label{eq:w_sigma_equivalence1}
\bSigma^{(k+1)} = \sum_{i=1}^N w_i^{(k)}\bx_i\bx_i^\top,
\end{equation}
then \eqref{eq:w_update} implies $w_i^{(k+1)} = \frac{D}{N \bx_i^\top(\bSigma^{(k+1)})^{-1}\bx_i}$. Therefore,
\begin{equation}\label{eq:w_sigma_equivalence2}
\bSigma^{(k+2)} = \sum_{i=1}^N w_i^{(k+1)}\bx_i\bx_i^\top = \frac{D}{N} \sum_{i=1}^N \frac{\bx_i\bx_i^\top}{\bx_i^\top(\bSigma^{(k+1)})^{-1}\bx_i},
\end{equation}
which is proportional to the update of $\boldsymbol{\Sigma}$ in \eqref{eq:tme_iteration}, reflecting the scale-invariance of the TME objective. Note that while the iteration in \eqref{eq:tme_iteration} constrains the trace to 1, this MM formulation naturally yields a sequence where the trace converges to $D$; however, due to the scale-invariance of the TME objective $F$, these two algorithmic paths are identical up to this fixed scalar multiple.

Next, we provide a sufficient condition for a nonsingular matrix to satisfy the fixed-point equation of the TME update. This proposition characterizes the limiting points of the sequence $\bSigma^{(k)}$ in our setting. 
Previous convergence analysis by \cite[Theorem 2]{Kent_Tyler88} assumed that $D/N < 1$ and showed that the TME algorithm converges to a nonsingular matrix satisfying the same fixed-point relationship. Furthermore, \cite[Theorem 1.1]{zhang2016robust} assumed the conditions of \Cref{thm:zhang16} (in particular, $D/N > 1$) and showed that the sequence converges to a singular matrix whose range is $L_*$. Our result below holds for a nonsingular matrix and is independent of the ratio $D/N$.

Next, we provide a sufficient condition for a nonsingular matrix to satisfy the fixed-point equation of the TME update. This proposition characterizes the limiting points of the sequence $\bSigma^{(k)}$ if it is nonsingular. 
Previous convergence analysis \cite[Theorem 2]{Kent_Tyler88} assumed that $\dssnr < 1$ and showed that the TME algorithm converges to a nonsingular matrix satisfying the same fixed-point relationship. Furthermore, \cite[Theorem 1.1]{zhang2016robust} assumed the conditions of \Cref{thm:zhang16} (in particular, $\dssnr > 1$) and showed that the sequence converges to a singular matrix whose range is $L_*$. Our result below holds for a nonsingular matrix and is independent of the ratio $\dssnr$.

\begin{proposition}\label{prop:nonsingular_fixed_point}
If $\hat{\bSigma}\in S_{++}(D)$ satisfies $F\!\left(T(\hat{\bSigma})\right) = F(\hat{\bSigma})$, then
\begin{equation}\label{eq:zhang16_extend1}
\hat{\bSigma} = \frac{D}{N} \sum_{\bx\in\calX} \frac{\bx\bx^\top}{\bx^\top\hat{\bSigma}^{-1}\bx}.
\end{equation}
Furthermore, if a subsequence of $\{\bSigma^{(k)}\}_{k \geq 0}$ converges to a nonsingular matrix $\hat{\bSigma}$, then $\hat{\bSigma}$ satisfies $F\!\left(T(\hat{\bSigma})\right) = F(\hat{\bSigma})$, and thus \eqref{eq:zhang16_extend1} holds.
\end{proposition}

\subsection{Proof of \Cref{prop:nonsingular_fixed_point}}
\label{sec:nonsingular_fixed_point}
Let $\hat{\bSigma}$ be an arbitrary matrix in $S_{++}$. Using the symmetry and invertibility of $\hat{\bSigma}$ alongside the matrix derivative formulas (57) and (61) of \cite{IMM2012-03274}, the derivative of the surrogate function $G$ with respect to its first argument evaluates to:
\begin{equation}\label{eq:derivative_G}
\frac{d}{d\bSigma} G(\bSigma,\hat{\bSigma}) = -\frac{1}{N} \sum_{\bx \in \calX} \frac{\bSigma^{-1}\bx \bx^\top\bSigma^{-1}}{\bx^\top \hat{\bSigma}^{-1} \bx} + \frac{1}{D} \bSigma^{-1} = \bSigma^{-1} \left( -\frac{1}{N} \sum_{\bx \in \calX} \frac{\bx \bx^\top}{\bx^\top \hat{\bSigma}^{-1} \bx} + \frac{1}{D} \bSigma \right) \bSigma^{-1}.
\end{equation}
This derivative evaluates to zero at $\bSigma = \hat{\bSigma}$ if and only if \eqref{eq:zhang16_extend1} holds. 

To prove the first statement of the proposition, we assume $F\!\left(T(\hat{\bSigma})\right) = F(\hat{\bSigma})$ and suppose, for the sake of contradiction, that \eqref{eq:zhang16_extend1} does not hold. Consequently, the derivative at $\hat{\bSigma}$ is non-zero, which implies that $\hat{\bSigma}$ cannot be a global minimum of the function $\bSigma \mapsto G(\bSigma,\hat{\bSigma})$. 
Therefore, letting $\widetilde{\bSigma} = \argmin_{\bSigma\in S_{++}} G(\bSigma,\hat{\bSigma})$ and recalling that $F\!\left(T(\hat{\bSigma})\right) = F(\widetilde{\bSigma})$ due to the scale-invariance of the objective $F$, we have
$$F\!\left(T(\hat{\bSigma})\right) = F(\widetilde{\bSigma}) \leq G(\widetilde{\bSigma},\hat{\bSigma}) = \min_{{\bSigma}\in S_{++}} G(\bSigma,\hat{\bSigma}) < G(\hat{\bSigma},\hat{\bSigma}) = F(\hat{\bSigma}),$$
which directly contradicts our initial assumption that $F\!\left(T(\hat{\bSigma})\right) = F\!\left(\hat{\bSigma}\right)$. This establishes the first statement.

For the second statement, assume that a subsequence of $\{\bSigma^{(k)}\}_{k \geq 0}$ converges to a nonsingular matrix $\hat{\bSigma}$. We note from \eqref{eq:tme_iteration} (and the fact that $\calX$ spans $\mathbb{R}^D$, as shown in Lemma \ref{lem:implicit_dimension}) that for each $k \geq 0$, $\bSigma^{(k)} \in S_{++}$. Thus, the limiting matrix must lie in the positive semi-definite cone $S_+$. Due to its assumed non-singularity, we strictly have $\hat{\bSigma} \in S_{++}$, ensuring both $F(\hat{\bSigma})$ and $T(\hat{\bSigma})$ are well-defined.

By continuity, the objective values along this subsequence converge to $F(\hat{\bSigma})$. Furthermore, by the monotonicity of $F$ established in \eqref{eq:monotone}, the objective values of the full sequence, $\lim_{k\to\infty} F(\bSigma^{(k)})$, must also converge to this exact same limit. As a result, the difference $F(\bSigma^{(k)}) - F(\bSigma^{(k+1)}) = F(\bSigma^{(k)}) - F(T(\bSigma^{(k)}))$ converges to zero. 
Combining this with the continuity of both $F$ and $T$ over the positive definite cone $S_{++}$ implies $F\!\left(T(\hat{\bSigma})\right) = F(\hat{\bSigma})$. Finally, by our first statement, this guarantees that \eqref{eq:zhang16_extend1} holds. \qed

\section{Proof of Theorem \ref{thm:zhang16_extend} for the Case of $\boldsymbol{\dssnr=1}$}
\label{sec:proof_thm_zhang16_extend}

The proof proceeds in three steps. In the first step, we show that at each iterate \( \bSigma^{(k)} \) admits a decomposition into two
components $\bSigma^{(k)}_{\mathrm{in}}$ and $\bSigma^{(k)}_{\mathrm{out}}$.  The first component $\bSigma^{(k)}_{\mathrm{in}}$ has a range contained in \( L_* \) and when $k\rightarrow\infty$, its magnitude is bounded from below by a positive number. On the other hand, the magnitude of $\bSigma^{(k)}_{\mathrm{out}}$ converges to zero as $k\rightarrow\infty$. Moreover, the condition numbers of both components are asymptotically bounded. 

In the second step, we show that the TME objective function $F$ in \eqref{eq:TME_objective} and the TME update transformation $T$ in \eqref{eq:tme_iteration} can also be decomposed into two parts. The first component depends on $\bSigma^{(k)}_{\mathrm{in}}$ and the set of inliers, and the second component depends on $\bSigma^{(k)}_{\mathrm{out}}$ and the set of outliers. 

In the third step, we show that \( \bSigma^{(k)} \) converges to the TME estimator computed from the inliers. The proof critically relies on a key consequence of the majorization--minimization principle---namely, both 
\( F(\bSigma^{(k)}) \) and $\tilde{F}(\bw^{(k)})$ are nonincreasing over the iterations.

\textbf{Step 1: A decomposition of $\bSigma^{(k)}$ and its properties.}
For $k\geq 1$, define \( \bSigma^{(k)}_{\mathrm{in}} \) and \( \bSigma^{(k)}_{\mathrm{out}} \) by
\[
\bSigma^{(k)}_{\mathrm{in}}
=
\sum_{\bx \in \calX_{\mathrm{in}}}
\frac{\bx \bx^\top}{\bx^\top (\bSigma^{(k-1)})^{-1} \bx},
\qquad
\bSigma^{(k)}_{\mathrm{out}}
=
\sum_{\bx \in \calX_{\mathrm{out}}}
\frac{\bx \bx^\top}{\bx^\top (\bSigma^{(k-1)})^{-1} \bx}.
\]
We note that 
\( \bSigma^{(k)}_{\mathrm{in}} \) and \( \bSigma^{(k)}_{\mathrm{out}} \) provide a natural
decomposition of \( \bSigma^{(k)} \), namely,
\begin{equation}
\label{eq:decompose_Sigmak}
\bSigma^{(k)}
=
\frac{\bSigma^{(k)}_{\mathrm{in}}+\bSigma^{(k)}_{\mathrm{out}}}
{\tr\!\bigl(\bSigma^{(k)}_{\mathrm{in}}+\bSigma^{(k)}_{\mathrm{out}}\bigr)}.
\end{equation}

The next lemma shows that \( \bSigma^{(k)}_{\mathrm{in}} \) and \( \bSigma^{(k)}_{\mathrm{out}} \) are  asymptotically well-conditioned and that the ``size'' of $\bSigma^{(k)}_{\mathrm{out}}$ is asymptotically diminishing compared to that of $\bSigma^{(k)}_{\mathrm{in}}$. 
\begin{lemma}\label{lemma:step1}
Both $\bSigma^{(k)}_{\mathrm{in}}$ and $\bSigma^{(k)}_{\mathrm{out}}$ are asymptotically well-conditioned in the sense that there exists a constant \( c > 0 \),  depending on $\calX$ and the initialization $\bSigma^{(0)}$, such that 
\begin{align}\label{eq:step1_conclusion1}
\liminf_{k\rightarrow\infty}\frac{\sigma_d\!\left(\bSigma^{(k)}_{\mathrm{in}}\right)}{\sigma_1\!\left(\bSigma^{(k)}_{\mathrm{in}}\right)}
&\geq
c\, , 
\\\label{eq:step1_conclusion2}
\liminf_{k\rightarrow\infty}\frac{\sigma_{D-d}\!\left(
[\bSigma^{(k)}_{\mathrm{out}}]_{L_*^\perp,L_*^\perp}
\right)}{\sigma_1\!\left(\bSigma^{(k)}_{\mathrm{out}}\right)}
&\geq
c. 
\end{align}
In addition, the relative size of $\bSigma^{(k)}_{\mathrm{out}}$ compared with $\bSigma^{(k)}_{\mathrm{in}}$ is asymptotically diminishing:
\begin{align}
\lim_{k \to \infty}
\frac{\sigma_1\!\left(\bSigma^{(k)}_{\mathrm{out}}\right)}
{\sigma_d\!\left(\bSigma^{(k)}_{\mathrm{in}}\right)}
&= 0.\label{eq:step1_conclusion3}
\end{align}
\end{lemma}

Combining \eqref{eq:decompose_Sigmak} with the observation 
$\bSigma^{(k)}_{\mathrm{in}} = \bP_{L_*} \bSigma^{(k)}_{\mathrm{in}} \bP_{L_*}$, noting that $\|\bA - \bP_{L_*} \bA \bP_{L_*}\|_F \leq \|\bA\|_F$ for any symmetric matrix $\bA$ (using it with $\bA = \bSigma_{\mathrm{out}}^{(k)}$), and applying a basic bound of the Frobenius norm yield
$$\|\bSigma^{(k)}-\bP_{L_*} \bSigma^{(k)} \bP_{L_*}\|_F
\leq \|\bSigma^{(k)}_{\mathrm{out}}\|_F/\tr\!\bigl(\bSigma^{(k)}_{\mathrm{in}}+\bSigma^{(k)}_{\mathrm{out}}\bigr)
\leq D\sigma_1(\bSigma^{(k)}_{\mathrm{out}})/d\sigma_d\!\left(\bSigma^{(k)}_{\mathrm{in}}\right).$$ 
The last equation and \eqref{eq:step1_conclusion3} imply that $\lim_{k\rightarrow\infty}\|\bSigma^{(k)}-\bP_{L_*} \bSigma^{(k)} \bP_{L_*}\|_F=0$, and thus     \begin{equation}\lim_{k\rightarrow\infty}\bSigma^{(k)}=\lim_{k\rightarrow\infty}\bP_{L_*} \bSigma^{(k)} \bP_{L_*}\label{eq:Sigmak_PSigmak}.\end{equation} 
We recall that to prove \Cref{thm:zhang16_extend} we need to show that $\lim_{k\to\infty}\bSigma^{(k)}=\bU_{L_*}\bSigma_{\mathrm{in},*}\bU_{L_*}^\top$, and by \eqref{eq:Sigmak_PSigmak}, it is equivalent to verifying $\lim_{k\to\infty}\bP_{L_*} \bSigma^{(k)} \bP_{L_*}=\bU_{L_*}\bSigma_{\mathrm{in},*}\bU_{L_*}^\top$. Since $\bP_{L_*} = \bU_{L_*} \bU_{L_*}^\top$ this is also equivalent to proving
\begin{equation}\label{eq:step1_implication}
\lim_{k\rightarrow\infty} [\bSigma^{(k)}]_{L_*,L_*}=\bSigma_{\mathrm{in},*},
\end{equation}
where we use the notation introduced in \eqref{eq:def_block_components} with $L_1=L_2=L_*$. 

\textbf{Step 2: A decomposition of the TME objective
and the corresponding iteration.}
We denote the TME objective in \eqref{eq:TME_objective} and the iterative
operator in \eqref{eq:tme_iteration} by \( F_{\calX} \) and \( T_{\calX} \), respectively, emphasizing
their dependence on the underlying dataset \( \calX \). We will also apply these operators to 
\begin{equation}
\label{eq:x_in_out_def} \tilde{\calX}_{\mathrm{in}}:=\{\bU_{L_*}^\top\bx: \bx\in\calX_{\mathrm{in}}\} \ \text{ and } \  \widehat{\calX}_{\mathrm{out}}:=\{\bU_{L_*^\perp}^\top\bx: \bx\in\calX_{\mathrm{out}}\}.    
\end{equation}
Then we apply the decomposition of $\bSigma$ presented in \eqref{eq:decompose_Sigmak} and its properties  \eqref{eq:step1_conclusion1}-\eqref{eq:step1_conclusion3} to show that both the
objective and the iterative operator admit an approximate decomposition into inlier and outlier
components. We formulate this decomposition in \Cref{lemma:step2}.  
We note that properties  \eqref{eq:step1_conclusion1}-\eqref{eq:step1_conclusion3} are independent of any scaling of the matrix $\bSigma$ and we thus ignore the numerator in \eqref{eq:decompose_Sigmak} when formulating this lemma. Furthermore, we avoid the asymptotic setting of properties  \eqref{eq:step1_conclusion1}-\eqref{eq:step1_conclusion3}. Instead of assuming \eqref{eq:step1_conclusion1} and \eqref{eq:step1_conclusion2}, we assume that their  lower bounds hold for $\bSigma^{(k)}$ with a sufficiently large $k$. Similarly, instead of directly using \eqref{eq:step1_conclusion3}, we work with the ratio in the left-hand side (LHS) of \eqref{eq:step1_conclusion3} and assume it is sufficiently small (we refer to it as $\epsilon$).

\begin{lemma}\label{lemma:step2}
Suppose that \( \bSigma \) admits a decomposition
\[
\bSigma = \bSigma_{\mathrm{in}} + \bSigma_{\mathrm{out}},
\]
where \( \bSigma_{\mathrm{in}} \) and \( \bSigma_{\mathrm{out}} \) are positive semidefinite matrices,
the range of \( \bSigma_{\mathrm{in}} \) is contained in \( L_* \), and 
there exists $c>0$ such that
\begin{equation}\label{eq:lemma2_assumption}
\frac{\sigma_d\!\left(\bSigma_{\mathrm{in}}\right)}{\sigma_1\!\left(\bSigma_{\mathrm{in}}\right)}
\geq
c,\,\,\frac{\sigma_{D-d}\!\left(
[\bSigma_{\mathrm{out}}]_{L_*^\perp,L_*^\perp}
\right)}{\sigma_1\!\left(\bSigma_{\mathrm{out}}\right)}
\geq
c.
\end{equation} 
Define
\[
\epsilon := \frac{\sigma_1(\bSigma_{\mathrm{out}})}{\sigma_d(\bSigma_{\mathrm{in}})}.
\]
and assume further that $\epsilon\leq 1$. Then there exists a constant $C=C(\calX,L_*)$ such that the following approximate decompositions hold: \begin{equation}\label{eq:objective_decomposition}
\left|F_{\calX}(\bSigma)
-
\left(\frac{n_1}{N}
F_{\tilde{\calX}_{\mathrm{in}}}
\!\left(
[\bSigma]_{L_*,L_*}
\right)
+
\frac{n_0}{N}
F_{\widehat{\calX}_{\mathrm{out}}}
\!\left(
[\bSigma]_{L_*^\perp,L_*^\perp}
\right)\right)\right|\leq C \epsilon,
\end{equation}
\begin{align}\label{eq:operator_decomposition}
\left\|\frac{
[T_{\calX}(\bSigma)]_{L_*,L_*}
}{
\tr\!\left(
[T_{\calX}(\bSigma)]_{L_*,L_*}
\right)
}
-
T_{\tilde{\calX}_{\mathrm{in}}}
\!\left(
[\bSigma]_{L_*,L_*}
\right)\right\|_F &\leq C \epsilon, \\
\left\|\frac{
[T_{\calX}(\bSigma)]_{L_*^\perp,L_*^\perp}
}{
\tr\!\left(
[T_{\calX}(\bSigma)]_{L_*^\perp,L_*^\perp}
\right)
}
-
T_{\widehat{\calX}_{\mathrm{out}}}
\!\left(
[\bSigma]_{L_*^\perp,L_*^\perp}
\right)\right\|_F &\leq 
C \epsilon.
\label{eq:operator_decomposition2}
\end{align}
\end{lemma}

\textbf{Step 3: Final arguments that \( \boldsymbol{\bSigma^{(k)}} \) converges to the TME estimator
\( \boldsymbol{\bU_{L_*} \bSigma_{\mathrm{in},*} \bU_{L_*}^\top} \).}  
Since the TME estimators with respect to the sets $\tilde{\calX}_{\mathrm{in}}$ and $\widehat{\calX}_{\mathrm{out}}$ exist (see \cref{lem:existence_uniqueness,lem:existence_uniqueness_outliers}, respectively), the objective functions $F_{\tilde{\calX}_{\mathrm{in}}}$ and $F_{\widehat{\calX}_{\mathrm{out}}}$ attain their global minimums and are therefore bounded from below. 
By \eqref{eq:step1_conclusion3} and \eqref{eq:objective_decomposition}, for a sufficiently large $K \in \nats$, the sequence $\{F_{\calX}(\bSigma^{(k)})\}_{k \geq K}$ is also bounded from below. This boundedness, combined with the monotonicity of the sequence $F_{\calX}(\bSigma^{(k)})$, implies that the sequence converges, and thus
\begin{equation}\label{eq:objective_converge}
\lim_{k\rightarrow\infty} \left( F_{\calX}(\bSigma^{(k)})-F_{\calX}(\bSigma^{(k+1)}) \right) = 0.
\end{equation}

For the rest of the proof we assume that $k \geq K$, without always explicitly mentioning this. We further adapt $K$ so that \eqref{eq:lemma2_assumption} holds for $\bSigma = \bSigma^{(k)}$ whenever $k \geq K$ and so that the following quantity
$$\epsilon_k=\frac{\sigma_1(\bSigma_{\mathrm{out}}^{(k)})}{\sigma_d(\bSigma_{\mathrm{in}}^{(k)})}$$ 
satisfies $\eps_k<1$ for any $k \geq K$. 
We can thus apply \Cref{lemma:step2}, where we explain the technical details below, to obtain for $k \geq K$: 
\begin{align}\nonumber
F_{\calX}(\bSigma^{(k+1)}) &\leq \frac{n_1}{N} F_{\tilde{\calX}_{\mathrm{in}}} \!\left( [\bSigma^{(k+1)}]_{L_*,L_*} \right) + \frac{n_0}{N} F_{\widehat{\calX}_{\mathrm{out}}} \!\left( [\bSigma^{(k+1)}]_{L_*^\perp,L_*^\perp} \right) + C \epsilon_{k+1} \\ \nonumber
&= \frac{n_1}{N} F_{\tilde{\calX}_{\mathrm{in}}} \!\left( \frac{[T_{\calX}(\bSigma^{(k)})]_{L_*,L_*}}{\tr\!\left([T_{\calX}(\bSigma^{(k)})]_{L_*,L_*}\right)} \right) + \frac{n_0}{N} F_{\widehat{\calX}_{\mathrm{out}}} \!\left( \frac{[T_{\calX}(\bSigma^{(k)})]_{L_*^\perp,L_*^\perp}}{\tr\!\left([T_{\calX}(\bSigma^{(k)})]_{L_*^\perp,L_*^\perp}\right)} \right) + C \epsilon_{k+1} \\ \nonumber
&\leq \frac{n_1}{N} F_{\tilde{\calX}_{\mathrm{in}}} \!\left( T_{\tilde{\calX}_{\mathrm{in}}} \!\left( [\bSigma^{(k)}]_{L_*,L_*} \right) \right) + \frac{n_0}{N} F_{\widehat{\calX}_{\mathrm{out}}} \!\left( T_{\widehat{\calX}_{\mathrm{out}}} \!\left( [\bSigma^{(k)}]_{L_*^\perp,L_*^\perp} \right) \right) + C(\epsilon_{k+1}+\epsilon_k) \\ \label{eq:inequality_decomposition}
&\leq \frac{n_1}{N} F_{\tilde{\calX}_{\mathrm{in}}} \!\left( [\bSigma^{(k)}]_{L_*,L_*} \right) + \frac{n_0}{N} F_{\widehat{\calX}_{\mathrm{out}}} \!\left( [\bSigma^{(k)}]_{L_*^\perp,L_*^\perp} \right) + C(\epsilon_{k+1}+\epsilon_k) \\ \nonumber
&\leq F_{\calX}(\bSigma^{(k)}) + C(\epsilon_{k+1}+\epsilon_k).
\end{align}
(Note that the generic constant $C > 0$ may change its value from line to line). Indeed, the first inequality follows from \eqref{eq:objective_decomposition}. The first equality follows from $\bSigma^{(k+1)}=T_{\calX}(\bSigma^{(k)})$ and the scale-invariance of $F$. The second inequality follows from \eqref{eq:operator_decomposition} and the Lipschitz continuity of $F_{\tilde{\calX}_{\mathrm{in}}}$ and $F_{\widehat{\mathcal{X}}_{\mathrm{out}}}$ on the set of well-conditioned matrices. We remark that this Lipschitz continuity follows directly from observing the derivative of $F$, which takes a form similar to that of $G$ as specified in \eqref{eq:derivative_G}. 

To complete this argument, we next verify that the matrices we use are indeed well-conditioned. We first note that by \eqref{eq:decompose_Sigmak} and \eqref{eq:step1_conclusion2},  $[{\bSigma}^{(k)}]_{L_*^\perp,L_*^\perp}$ is well-conditioned:
\[
\frac{\sigma_1([{\bSigma}^{(k)}]_{L_*^\perp,L_*^\perp})}{\sigma_d([{\bSigma}^{(k)}]_{L_*^\perp,L_*^\perp})}
=
\frac{\sigma_1([{\bSigma}^{(k)}_{\mathrm{out}}]_{L_*^\perp,L_*^\perp})}{\sigma_d([{\bSigma}^{(k)}_{\mathrm{out}}]_{L_*^\perp,L_*^\perp})}\leq 1/c,
\]
where the above constant $c$ is slightly lower than the one in \eqref{eq:step1_conclusion2}. 
Next, we show that $[{\bSigma}^{(k)}]_{L_*,L_*}$ is also well-conditioned for sufficiently large $k$, while using again a slightly lower value of $c$: 
\begin{equation}\label{eq:Sigma_L_nonsingular}
\frac{\sigma_1([{\bSigma}^{(k)}]_{L_*,L_*})}{\sigma_d([{\bSigma}^{(k)}]_{L_*,L_*})}
\leq \frac{\sigma_1([{\bSigma}^{(k)}_{\mathrm{in}}]_{L_*,L_*})+\sigma_1([{\bSigma}^{(k)}_{\mathrm{out}}]_{L_*,L_*})}{\sigma_d([{\bSigma}^{(k)}_{\mathrm{in}}]_{L_*,L_*})}\leq 1/c+\epsilon_k\rightarrow 1/c.
\end{equation}
To obtain the first inequality, we substitute the decomposition from \eqref{eq:decompose_Sigmak} into both the numerator and denominator. The scalar normalization factor $1/\tr(\bSigma_{\mathrm{in}}^{(k)}+\bSigma_{\mathrm{out}}^{(k)})$ cancels out, allowing us to bound the ratio using the spectral norms of the unnormalized components. 
The bound on the numerator follows from the fact that $\sigma_1(\cdot)$ is a norm and the bound on the denominator, that is, the inequality $\sigma_d([{\bSigma}^{(k)}_{\mathrm{in}}]_{L_*,L_*}+[{\bSigma}^{(k)}_{\mathrm{out}}]_{L_*,L_*})\geq \sigma_d([{\bSigma}^{(k)}_{\mathrm{in}}]_{L_*,L_*})$, follows from Weyl's Monotonicity Theorem \cite[Corollary III.2.3]{bhatia1996matrix} and the fact that ${\bSigma}^{(k)}_{\mathrm{out}} \in S_{+}(D)$. 
The second inequality follows from \eqref{eq:step1_conclusion1} and the definition of $\epsilon_k$, and the limit follows from \eqref{eq:step1_conclusion3}.

The third inequality of \eqref{eq:inequality_decomposition} follows from the argument in \eqref{eq:monotone}, which can be written for our case as
\begin{equation}\label{eq:TME_separate_monotonicity}
F_{\tilde{\calX}_{\mathrm{in}}}
\!\left(
T_{\tilde{\calX}_{\mathrm{in}}}(\bSigma)
\right)
\leq
F_{\tilde{\calX}_{\mathrm{in}}}(\bSigma),
\qquad
F_{\widehat{\calX}_{\mathrm{out}}}
\!\left(
T_{\widehat{\calX}_{\mathrm{out}}}(\bSigma)
\right)
\leq
F_{\widehat{\calX}_{\mathrm{out}}}(\bSigma).
\end{equation}
Lastly, the fourth inequality of \eqref{eq:inequality_decomposition} follows from \eqref{eq:operator_decomposition2}. 

Combining \eqref{eq:objective_converge}    and the fact that
\( \lim_{k \to \infty} \epsilon_k = 0 \), we see that, asymptotically, the
inequality in  \eqref{eq:inequality_decomposition} becomes an equality. That is,
\begin{align*}
\lim_{k\rightarrow \infty}\frac{n_1}{N}&\left(F_{\tilde{\calX}_{\mathrm{in}}}
\!\left(T_{\tilde{\calX}_{\mathrm{in}}}(
[\bSigma^{(k)}]_{L_*,L_*}
)\right)-F_{\tilde{\calX}_{\mathrm{in}}}
\!\left(
[\bSigma^{(k)}]_{L_*,L_*}
\right)\right)
\\&+
\frac{n_0}{N}
\left(F_{\widehat{\calX}_{\mathrm{out}}}
\!\left(T_{\widehat{\calX}_{\mathrm{out}}}(
[\bSigma^{(k)}]_{L_*^\perp,L_*^\perp})
\right)-F_{\widehat{\calX}_{\mathrm{out}}}
\!\left(
[\bSigma^{(k)}]_{L_*^\perp,L_*^\perp}
\right)\right)
=0.
\end{align*}
Combining the above equation with \eqref{eq:TME_separate_monotonicity} yields 
\[
\lim_{k \to \infty}
\Bigl[
F_{\tilde{\calX}_{\mathrm{in}}}
\!\left(
[\bSigma^{(k)}]_{L_*,L_*}\right)
-
F_{\tilde{\calX}_{\mathrm{in}}}
\!\left(
T_{\tilde{\calX}_{\mathrm{in}}}
\bigl([\bSigma^{(k)}]_{L_*,L_*}
\bigr)
\right)
\Bigr]
= 0.
\]

Now suppose that $\hat{\bSigma}$ is a limiting point of the sequence
\(
\{[\bSigma^{(k)}]_{L_*,L_*}\}_{k\geq 1},
\)
then by the above equation and the continuity of $F_{\tilde{\calX}_{\mathrm{in}}}$:
\[
F_{\tilde{\calX}_{\mathrm{in}}}(\hat{\bSigma})
-
F_{\tilde{\calX}_{\mathrm{in}}}
\!\left(
T_{\tilde{\calX}_{\mathrm{in}}}(\hat{\bSigma})
\right)
= 0.
\]
We recall that for sufficiently large $k$, any $[\bSigma^{(k)}]_{L_*,L_*}$ is uniformly well-conditioned (see \eqref{eq:Sigma_L_nonsingular}) and thus the limit $\hat{\bSigma}$ is nonsingular.  Therefore, \Cref{prop:nonsingular_fixed_point} implies that $\hat{\bSigma}$ satisfies the fixed-point condition in \eqref{eq:zhang16_extend1}.  We also note that the TME estimator $\bSigma_{\mathrm{in},*}$ needs to satisfy the same fixed-point condition. Indeed, this is obtained by differentiation similar to the one used to derive \eqref{eq:fixed_from_Ftilde} (where $\tilde{F}$ replaces $F$). 

Next, we claim that for the dataset $\tilde{\calX}_{\mathrm{in}}$, the solution of the fixed-point equation is unique and since we showed that both 
\( \hat{\bSigma} \) and \(\bSigma_{\mathrm{in},*} \) satisfy it we conclude that 
\( \hat{\bSigma} = \bSigma_{\mathrm{in},*} \). 
The application of \cref{item:convergence} and the assumption $\dssnr=1$ implies that any proper subspace $L$ of $L_*$ contains less than $\dim(L)\frac{N}{D} = \dim(L)\frac{n_1}{d}$ points. It thus follows from \cite[Theorem 1.2]{zhang2016robust} that the fixed point is unique. We thus showed that any limiting point of the sequence $[\bSigma^{(k)}]_{L_*,L_*}$, which we denoted by $\hat{\bSigma}$, equals $\bSigma_{\mathrm{in},*}$, and in particular, the limiting point is unique (if it exists). Its existence follows from the compactness of the set 
\(
\{[\bSigma^{(k)}]_{L_*,L_*}\}_{k \geq 1}.
\)
Indeed, this is a set of positive semidefinite matrices with trace at most \(1\). Therefore, 
\[
\lim_{k \to \infty} [\bSigma^{(k)}]_{L_*,L_*}
= \lim_{k \to \infty} \bU_{L_*}^\top \bSigma^{(k)} \bU_{L_*}
= \bSigma_{\mathrm{in},*}.
\]
We recall that the above equation implies 
\Cref{thm:zhang16_extend} (see the end of step 1).


\subsection{Proof of \Cref{lemma:step1}}

The proof proceeds by analyzing the sequence $\{\bw^{(k)}\}_{k\ge 1}$. It shows that there exists a constant $c>0$ such that the following three properties hold.

\begin{enumerate}
\item If $\bx_{i_1}, \bx_{i_2} \in L_*$, then
\begin{equation}\label{eq:lemma1_w1}
\liminf_{k\to\infty} \frac{w_{i_1}^{(k)}}{w_{i_2}^{(k)}} > c.
\end{equation}

\item If $\bx_{i_1}, \bx_{i_2} \in \R^D \setminus L_*$, then
\begin{equation}\label{eq:lemma1_w2}
\liminf_{k\to\infty} \frac{w_{i_1}^{(k)}}{w_{i_2}^{(k)}} > c.
\end{equation}

\item If $\bx_{i_1} \in L_*$ and $\bx_{i_2} \in \R^D \setminus L_*$, then
\begin{equation}\label{eq:lemma1_w3}
\lim_{k\to\infty} \frac{w_{i_2}^{(k)}}{w_{i_1}^{(k)}} =0 .
\end{equation}
\end{enumerate}

Before verifying these claims we will show that they imply \Cref{lemma:step1}, i.e., \eqref{eq:step1_conclusion1}-\eqref{eq:step1_conclusion3}.  We recall that  \eqref{eq:w_sigma_equivalence1} and \eqref{eq:w_sigma_equivalence2} implies that  
$\bSigma^{(k)}$ is a positive scalar multiple of $\sum_i w_i^{(k)} \bx_i \bx_i^\top$. 
Since \eqref{eq:step1_conclusion1}--\eqref{eq:step1_conclusion3} concern only ratios, 
this scaling factor does not affect the statements we aim to prove. As a result, without loss of generality, we assume for the remainder of the proof:
\[
\bSigma^{(k)} = \sum_i w_i^{(k)} \bx_i \bx_i^\top, \ \bSigma^{(k)}_{\mathrm{in}}=\sum_{\bx_i\in \calX_{\mathrm{in}}}w_{i}^{(k)}\bx_i\bx_i^\top, \ \text{ and } \  \bSigma^{(k)}_{\mathrm{out}}=\sum_{\bx_i\in \calX_{\mathrm{out}}}w_{i}^{(k)}\bx_i\bx_i^\top.
\]

We first conclude  \eqref{eq:step1_conclusion1}.  We note that
\[
\bSigma^{(k)}_{\mathrm{in}}=\sum_{\bx_i\in \calX_{\mathrm{in}}}w_{i}^{(k)}\bx_i\bx_i^\top\psdgeq \min_{\bx_i\in \calX_{\mathrm{in}}}w_{i}^{(k)}\sum_{\bx_i\in \calX_{\mathrm{in}}}\bx_i\bx_i^\top,
\]
and thus 
\begin{equation}
\label{eq:sigma_d_geq}
\sigma_d(\bSigma^{(k)}_{\mathrm{in}})\geq \min_{\bx_i\in \calX_{\mathrm{in}}}w_{i}^{(k)} \cdot \sigma_d\!\left(\sum_{\bx_i\in \calX_{\mathrm{in}}}\bx_i\bx_i^\top\right).     
\end{equation}
Similarly,
\[
\bSigma^{(k)}_{\mathrm{in}}=\sum_{\bx_i\in \calX_{\mathrm{in}}}w_{i}^{(k)}\bx_i\bx_i^\top\psdleq \max_{\bx_i\in \calX_{\mathrm{in}}}w_{i}^{(k)}\sum_{\bx_i\in \calX_{\mathrm{in}}}\bx_i\bx_i^\top,
\]
which implies
\begin{equation}
\label{eq:sigma_1_leq}
\sigma_1(\bSigma^{(k)}_{\mathrm{in}})\leq \max_{\bx_i\in \calX_{\mathrm{in}}}w_{i}^{(k)}\sigma_1(\sum_{\bx_i\in \calX_{\mathrm{in}}}\bx_i\bx_i^\top). 
\end{equation}
We further observe that \eqref{eq:lemma1_w1} yields
\[
\liminf_{k\rightarrow\infty}\frac{\min_{\bx_i\in \calX_{\mathrm{in}}}w_{i}^{(k)}}{\max_{\bx_i\in \calX_{\mathrm{in}}}w_{i}^{(k)}} > c.
\]
Combining \eqref{eq:sigma_d_geq}, \eqref{eq:sigma_1_leq} and the above equation, we conclude 
\eqref{eq:step1_conclusion1}:
\[
\liminf_{k\rightarrow\infty}\frac{\sigma_d\!\left(\bSigma^{(k)}_{\mathrm{in}}\right)}{ \sigma_1\!\left(\bSigma^{(k)}_{\mathrm{in}}\right)}\geq  \liminf_k\frac{\min_{\bx_i\in \calX_{\mathrm{in}}}w_{i}^{(k)}}{\max_{\bx_i\in \calX_{\mathrm{in}}}w_{i}^{(k)}} \frac{\sigma_d\left(\sum_{\bx_i\in \calX_{\mathrm{in}}}\bx_i\bx_i^\top \right)}{\sigma_1\left(\sum_{\bx_i\in \calX_{\mathrm{in}}}\bx_i\bx_i^\top \right)}>c \frac{\sigma_d\left(\sum_{\bx_i\in \calX_{\mathrm{in}}}\bx_i\bx_i^\top \right)}{\sigma_1\left(\sum_{\bx_i\in \calX_{\mathrm{in}}}\bx_i\bx_i^\top \right)}.
\]
This proves \eqref{eq:step1_conclusion1} with $c:=c \cdot {\sigma_d\left(\sum_{\bx_i\in \calX_{\mathrm{in}}}\bx_i\bx_i^\top \right)}/{\sigma_1\left(\sum_{\bx_i\in \calX_{\mathrm{in}}}\bx_i\bx_i^\top \right)}$, which only depends on the fixed data set $\calX$.

The proof of \eqref{eq:step1_conclusion2} is the same, where $\calX_{\mathrm{in}}$ is replaced with $\calX_{\mathrm{out}}$, and $\bSigma^{(k)}_{\mathrm{in}}$ is replaced with $\bSigma^{(k)}_{\mathrm{out}}$.

To prove \eqref{eq:step1_conclusion3}, we first apply \eqref{eq:lemma1_w3} to conclude  $\lim_{k\rightarrow\infty}\frac{\min_{\bx_i\in \calX_{\mathrm{out}}}w_{i}^{(k)}}{\max_{\bx_i\in \calX_{\mathrm{in}}}w_{i}^{(k)}}=0$. Applying the nonnegativity of singular values,   \eqref{eq:sigma_d_geq},  \eqref{eq:sigma_1_leq}, and the latter observation, yields \eqref{eq:step1_conclusion3}:
\[
0 \leq \lim_{k\rightarrow\infty}\frac{\sigma_1\!\left(\bSigma^{(k)}_{\mathrm{out}}\right)}
{\sigma_d\!\left(\bSigma^{(k)}_{\mathrm{in}}\right)}\leq \lim_{k\rightarrow\infty}\frac{\min_{\bx_i\in \calX_{\mathrm{out}}}w_{i}^{(k)}}{\max_{\bx_i\in \calX_{\mathrm{in}}}w_{i}^{(k)}}
\frac{\sigma_1(\sum_{\bx_i\in \calX_{\mathrm{out}}}\bx_i\bx_i^\top)}{\sigma_d(\sum_{\bx_i\in \calX_{\mathrm{out}}}\bx_i\bx_i^\top)}=0.
\]

It remains to prove \eqref{eq:lemma1_w1}-\eqref{eq:lemma1_w3}. We pursue this in two different steps. 

\textbf{Step 1 in proving \eqref{eq:lemma1_w1}--\eqref{eq:lemma1_w3}.} By passing to a subsequence indexed by $\calK$, we may assume that the ordering of the weights remains fixed. That is, there exists a permutation $(i_1,\ldots,i_N)$ of $(1,\ldots,N)$ such that
\begin{equation}\label{eq:subsequence_property1}
w_{i_1}^{(k)} \ge w_{i_2}^{(k)} \ge \cdots \ge w_{i_N}^{(k)}, \qquad \forall\, k \in \calK.
\end{equation}
Furthermore, by extracting a further subsequence if necessary, we can assume that the limits of all consecutive weight ratios exist:
\begin{equation}\label{eq:subsequence_property2}
\lim_{\substack{k\to\infty \\ k\in\calK}} \frac{w_{i_{j+1}}^{(k)}}{w_{i_j}^{(k)}} \quad \text{exists for all } 1 \le j \le N-1.
\end{equation}
The goal of this step is to prove: 
\begin{equation}\label{eq:lemma1_step1_conclusion}
\begin{aligned}
\lim_{\substack{k\to\infty \\ k\in\calK}} \frac{w_{i_2}^{(k)}}{w_{i_1}^{(k)}} \;&=\; 0 && \text{for all } \bx_{i_1} \in \calX_{\mathrm{in}} \text{ and } \bx_{i_2} \in \calX_{\mathrm{out}}, \\
c_0 \;\le\; \lim_{\substack{k\to\infty \\ k\in\calK}} \frac{w_{i_1}^{(k)}}{w_{i_2}^{(k)}} \;&\le\; \frac{1}{c_0} && \text{for any pair of inliers, or pair of outliers } (\bx_{i_1}, \bx_{i_2}).
\end{aligned}
\end{equation}

Because the weights are nonnegative and ordered according to \eqref{eq:subsequence_property1}, the limits of the consecutive ratios in \eqref{eq:subsequence_property2} must lie in the interval $[0,1]$. Based on these limits, we can partition the index set $\{1,\ldots,N\}$ into $T$ disjoint subsets, $\calI_1,\ldots,\calI_T$, grouping together indices whose weights scale at the same rate. This partition satisfies the following two properties:
\begin{enumerate}
    \item 
    For any two indices $i_1,i_2 \in \calI_\ell$, where $1 \leq \ell \leq T$, their ratio converges to a strictly positive, finite constant. That is, there exists a constant $c_0 \in (0,1]$ such that
    \[
    c_0 \;\le\; \lim_{\substack{k\to\infty \\ k\in\calK}} \frac{w_{i_1}^{(k)}}{w_{i_2}^{(k)}} \;\le\; \frac{1}{c_0}.
    \]
    
    \item
    For any two indices $i_1 \in \calI_\ell$ and $i_2 \in \calI_m$ belonging to different subsets ($1 \leq \ell \neq m \leq T$), their ratio diverges or vanishes:
    \[
    \lim_{\substack{k\to\infty \\ k\in\calK}} \frac{w_{i_1}^{(k)}}{w_{i_2}^{(k)}} \in \{0,\infty\}.
    \]
\end{enumerate}

Equivalently, for each cluster $\calI_\ell$, we can define a representative scale sequence $\epsilon_\ell^{(k)}$ such that the weights $w_i^{(k)}$ for $i \in \calI_\ell$ are strictly bounded by a multiple of $\epsilon_\ell^{(k)}$. Specifically, we have
\[
c_0 \epsilon_\ell^{(k)} \leq w_i^{(k)} \leq \frac{1}{c_0}\epsilon_\ell^{(k)}, \qquad \forall\, i \in \calI_\ell.
\]
Moreover, these representative sequences $\{\epsilon_\ell^{(k)}\}_{\ell=1}^\top$ satisfy a strict asymptotic hierarchy,
\[
\epsilon_1^{(k)} \gg \epsilon_2^{(k)} \gg \cdots \gg \epsilon_T^{(k)},
\]
meaning that
\begin{equation}\label{eq:epsilon_ratio}
\lim_{\substack{k\to\infty \\ k\in\calK}} \frac{\epsilon_\ell^{(k)}}{\epsilon_{\ell+1}^{(k)}} = \infty \quad \text{for all } \ell = 1,\ldots,T-1.
\end{equation}

Let $L_\ell$ denote the subspace spanned by the vectors associated with the first $\ell$ clusters, and let $D_\ell$ be its dimension:
\[
L_\ell \;=\; \Sp\bigl(\{\bx_i\}_{i \in \calI_1 \cup \cdots \cup \calI_\ell}\bigr),
\qquad
D_\ell \;=\; \dim(L_\ell).
\]
By convention, we set $D_0 = 0$. Note that $D_\ell < D_{\ell+1}$ strictly. Indeed, if the vectors in $\calI_{\ell+1}$ were fully contained in $L_\ell$, they would share the same asymptotic weight order as the prior clusters, contradicting the asymptotic partition.

Because of the asymptotic hierarchy of the weights, the matrix 
$
\sum_{i=1}^N w_i^{(k)} \bx_i \bx_i^\top
$
has exactly $D_\ell - D_{\ell-1}$ eigenvalues that scale on the order of $\epsilon_\ell^{(k)}$. Consequently, its log-determinant can be expanded as the sum of the logs of its eigenvalues, $\sigma_d$:
\[
\log \det\!\left( \sum_{i=1}^N w_i^{(k)} \bx_i \bx_i^\top \right) = \sum_{d=1}^D \log\sigma_d\!\left( \sum_{i=1}^N w_i^{(k)} \bx_i \bx_i^\top \right).
\]
Asymptotically, this behaves as
\begin{equation}\label{eq:asymptoticorder1}
\sum_{\ell=1}^T (D_\ell - D_{\ell-1}) \log \epsilon_\ell^{(k)}
\;=\;
\sum_{\ell=1}^{T-1} D_\ell \log \frac{\epsilon_\ell^{(k)}}{\epsilon_{\ell+1}^{(k)}}
\;+\;
D_T \log \epsilon_{T}^{(k)}.
\end{equation}

Similarly, let $N_\ell$ be the total number of points in the first $\ell$ clusters: 
$N_\ell \;=\; \bigl| \calI_1 \cup \cdots \cup \calI_\ell \bigr|$,
with $N_0 = 0$. Then the sum of the log-weights, 
$\sum_{i=1}^N \log w_i^{(k)}$, 
has the asymptotic order
\begin{equation}\label{eq:asymptoticorder2}
\sum_{\ell=1}^T (N_\ell - N_{\ell-1}) \log \epsilon_\ell^{(k)}
\;=\;
\sum_{\ell=1}^{T-1} N_\ell \log \frac{\epsilon_\ell^{(k)}}{\epsilon_{\ell+1}^{(k)}}
\;+\;
N_T \log \epsilon_{T}^{(k)}.
\end{equation}

Finally, we combine \eqref{eq:asymptoticorder1} and \eqref{eq:asymptoticorder2}. Noting that $D_T = D$ and $N_T = N$ by construction, we find that $\tilde{F}(\bw^{(k)})$, defined in \eqref{eq:TME_objective2}, is of the asymptotic order
\begin{equation}
\label{eq:asymptotic_sum_log}
\sum_{\ell=1}^{T-1} \left(\frac{D_\ell}{D}-\frac{N_\ell}{N}\right) \log \frac{\epsilon_\ell^{(k)}}{\epsilon_{\ell+1}^{(k)}}.
\end{equation}

Since $\tilde{F}(\bw^{(k)})$ is nonincreasing, it must remain bounded as $k \to \infty$. By \cref{item:convergence} in \Cref{thm:zhang16_extend}, we have
\[
\frac{D_\ell}{D} \;\ge\; \frac{N_\ell}{N} \qquad \text{for all } 1 \le \ell \le T,
\]
with equality holding if and only if $L_\ell = \R^D$ or $L_\ell = L_*$. Recalling from \eqref{eq:epsilon_ratio} that the log-ratios diverge to infinity, the asymptotic sum for $\tilde{F}(\bw^{(k)})$ will diverge to $+\infty$ unless every coefficient $(D_\ell/D - N_\ell/N)$ vanishes. This formally restricts us to two scenarios, though we will show that the first one is impossible. 

\medskip
\noindent
\textbf{Case 1:}   $T = 1$, so that the asymptotic sum in \eqref{eq:asymptotic_sum_log} is empty. In this case, all weights $\{w_i^{(k)}\}_{i=1}^N$ share the same asymptotic order for $k \in \mathcal{K}$. This corresponds to the scenario where $L_1 = \mathbb{R}^D$. We show this case is impossible. 

 Combining the fact that all $w_i^{(k)}$ are of the same asymptotic order with the assumption $\dssnr=1$, we conclude that the limit of $\bSigma^{(k+1)}$ as $k$ approaches infinity, where $k\in\calK$, is nonsingular.  \Cref{prop:nonsingular_fixed_point} then implies that this limit is a fixed point. 
This leads to a contradiction. Indeed, denoting this fixed point by $\hat{\bSigma}$,   Proposition~\ref{prop:nonsingular_fixed_point} implies
\begin{equation}\label{eq:case1_1}
\sum_{\bx\in\calX}\frac{(\hat{\bSigma}^{-1/2}\bx)(\hat{\bSigma}^{-1/2}\bx)^\top}{\|\hat{\bSigma}^{-1/2}\bx\|^2}
=\frac{N}{D}\bI.
\end{equation}

Since for every inlier  $\bx\in{\calX}_{\mathrm{in}}$, the vector $\hat{\bSigma}^{-1/2}\bx$ lies in a $d$-dimensional subspace (denoted by $\hat{L}$), it follows  from \eqref{eq:case1_1} that
\begin{equation}\label{eq:case1_3}
\sum_{\bx\in{\calX}_{\mathrm{in}}}
\frac{(\hat{\bSigma}^{-1/2}\bx)(\hat{\bSigma}^{-1/2}\bx)^\top}
{\|\hat{\bSigma}^{-1/2}\bx\|^2}
\psdleq \left\|\sum_{\bx\in\calX}\frac{(\hat{\bSigma}^{-1/2}\bx)(\hat{\bSigma}^{-1/2}\bx)^\top}{\|\hat{\bSigma}^{-1/2}\bx\|^2}\right\|\bP_{\hat{L}}=\frac{N}{D}\bP_{\hat{L}}.
\end{equation}
Moreover,
\begin{equation}\label{eq:case1_2}
\tr\!\left(\sum_{\bx\in{\calX}_{\mathrm{in}}}
\frac{(\hat{\bSigma}^{-1/2}\bx)(\hat{\bSigma}^{-1/2}\bx)^\top}
{\|\hat{\bSigma}^{-1/2}\bx\|^2}\right)=\sum_{\bx\in{\calX}_{\mathrm{in}}}\tr\!\left(
\frac{(\hat{\bSigma}^{-1/2}\bx)(\hat{\bSigma}^{-1/2}\bx)^\top}
{\|\hat{\bSigma}^{-1/2}\bx\|^2}\right)
=|{\calX}_{\mathrm{in}}|
\geq \frac{N}{D}d.
\end{equation}
Combining \eqref{eq:case1_3}, \eqref{eq:case1_2}, and $\dim(\hat{L})=d$, we have 
\[
\sum_{\bx\in{\calX}_{\mathrm{in}}}
\frac{(\hat{\bSigma}^{-1/2}\bx)(\hat{\bSigma}^{-1/2}\bx)^\top}
{\|\hat{\bSigma}^{-1/2}\bx\|^2}
=\frac{N}{D}\bP_{\hat{L}}.
\]
Consequently,
\[
\sum_{\bx\in{\calX}_{\mathrm{out}}}
\frac{(\hat{\bSigma}^{-1/2}\bx)(\hat{\bSigma}^{-1/2}\bx)^\top}
{\|\hat{\bSigma}^{-1/2}\bx\|^2}
=\frac{N}{D}\bI-\frac{N}{D}\bP_{\hat{L}}
=\frac{N}{D}\bP_{\hat{L}^\perp}.
\]
This implies that for all outliers $\bx\in{\calX}_{\mathrm{out}}$, $\hat{\bSigma}^{-1/2}\bx$ lie in a $(D-d)$-dimensional subspace, which also implies that the outliers themselves lie in a $(D-d)$-dimensional subspace, and contradicts \cref{item:convergence} of \Cref{thm:zhang16_extend} as  this subspace contains a fraction $(D-d)/D$ of points.

\medskip
\noindent
\textbf{Case 2:} $T = 2$. The weights partition into exactly two scales: $\ell=1$ represents the inliers with $L_1 = L_*$, and $\ell_2$ represents the outliers, where $L_2 = \R^D$.

All inlier weights $\{w_i^{(k)}\}_{i:\,\bx_i \in \calX_{\mathrm{in}}}$ share the same asymptotic order, all outlier weights $\{w_i^{(k)}\}_{i:\,\bx_i \in \calX_{\mathrm{out}}}$ share a strictly smaller asymptotic order, and the ratio of any outlier weight to any inlier weight vanishes. More precisely, it satisfies \eqref{eq:lemma1_step1_conclusion}. 

There are no other cases. Indeed, by the analysis before Case 1, we have that for all $1\leq \ell\leq T$, $L_\ell = \R^D$ or $L_\ell = L_*$, that is, $D_{\ell}=d$ or $D$. In addition, by definition, the sequence $D_\ell$ is strictly increasing. Therefore, case 1 happens when there does not exist $\ell$ such that $D_\ell=d$. Then $D_\ell=D$ for $1 \leq \ell \leq T$, and since $D_\ell$ is strictly increasing, we have $T=1$, $D_1=D$ and $L_1=\R^D$. 
Case 2 corresponds to the remaining case that there exists $\ell$ such that  $D_\ell=d$. Since $D_\ell$ is strictly increasing, we have $T=2$ with $D_1=d$ and $D_2=D$. 
Since only case 2 is valid, we established  \eqref{eq:lemma1_step1_conclusion}.

\textbf{Step 2 in proving \eqref{eq:lemma1_w1}--\eqref{eq:lemma1_w3}.}
Suppose, to the contrary, that \eqref{eq:lemma1_w1}--\eqref{eq:lemma1_w3} do not hold. Then there exists a general subsequence of $\{1,2,\cdots\}$ indexed by $\calK_1$ such that at least one of the following conditions is satisfied:
\begin{itemize}
\item For some inliers $\bx_{i_1}$ and $\bx_{i_2}$,
\[
\lim_{k\to\infty,\; k\in \calK_1} \frac{w_{i_1}^{(k)}}{w_{i_2}^{(k)}} = 0.
\]

\item For some outliers $\bx_{i_1}$ and $\bx_{i_2}$,
\[
\lim_{k\to\infty,\; k\in \calK_1} \frac{w_{i_1}^{(k)}}{w_{i_2}^{(k)}} = 0.
\]
\item For some inlier $\bx_{i_1}$ and some outlier $\bx_{i_2}$,
\[
\lim_{k\to\infty,\; k\in \calK_1} \frac{w_{i_2}^{(k)}}{w_{i_1}^{(k)}} >0.
\]
\end{itemize}

Since each $\bw^{(k)}$ is nonnegative, we may extract a further subsequence indexed by $\calK_2 \subset \calK_1$ such that, for all $k \in \calK_2$, properties \eqref{eq:subsequence_property1}--\eqref{eq:subsequence_property2} hold. By Step 1, this implies that \eqref{eq:lemma1_step1_conclusion} holds for $\calK_2$. Because $\calK_2$ is a subsequence of $\calK_1$, it must inherit the limit of whichever of the above three equations  $\calK_1$ satisfied. However, \eqref{eq:lemma1_step1_conclusion} directly contradicts all of these equations. This contradiction completes the proof of \eqref{eq:lemma1_w1}--\eqref{eq:lemma1_w3}.

\subsection{Proof of \Cref{lemma:step2}}

The proof will repeatedly make use of the block matrix inversion formula. Specifically, let
\begin{equation}
\label{eq:def_sigma_prime}
\bSigma'
\;=\;
\bigl(
[\bSigma]_{L_*,L_*}
-
[\bSigma]_{L_*,L_*^\perp}
[\bSigma]_{L_*^\perp,L_*^\perp}^{-1}
[\bSigma]_{L_*^\perp,L_*}
\bigr)^{-1}.
\end{equation}
Then the inverse of $\bSigma$ admits the block representation
\begin{equation}\label{eq:block_representation}
\bSigma^{-1}
=
\begin{pmatrix}
\bSigma'
&
-\bSigma'[\bSigma]_{L_*,L_*^\perp}[\bSigma]_{L_*^\perp,L_*^\perp}^{-1}
\\[0.6em]
-[\bSigma]_{L_*^\perp,L_*^\perp}^{-1}[\bSigma]_{L_*^\perp,L_*}\bSigma'
&
[\bSigma]_{L_*^\perp,L_*^\perp}^{-1}
+
[\bSigma]_{L_*^\perp,L_*^\perp}^{-1}
[\bSigma]_{L_*^\perp,L_*}\bSigma'
[\bSigma]_{L_*,L_*^\perp}[\bSigma]_{L_*^\perp,L_*^\perp}^{-1}
\end{pmatrix}.
\end{equation}

We first estimate the magnitude of the quantities above. 
Since both $F$ and $T$ are invariant to the scale of $\bSigma$, without loss of generality (WLOG), we may assume that $\sigma_d\!\left(\bSigma_{\mathrm{in}}\right)=1$. Then we have 
\begin{align}
&\sigma_1\!\left(\bSigma_{\mathrm{in}}\right)\leq \frac{1}{c},\,\,
\sigma_1(\bSigma_{\mathrm{out}})=\epsilon, \,\,\sigma_{D-d}\!\left(
[\bSigma_{\mathrm{out}}]_{L_*^\perp,L_*^\perp}
\right)\geq c\epsilon,  \label{eq:estimations_magnitude1}
\\
&\|[\bSigma]_{L_*^\perp,L_*}\|=\|[\bSigma_{\mathrm{out}}]_{L_*^\perp,L_*}\|\leq \sigma_1(\bSigma_{\mathrm{out}})=\epsilon, \label{eq:estimations_magnitude2}
\\
&\|[\bSigma]_{L_*^\perp,L_*^\perp}^{-1}\|=\|[\bSigma_{\mathrm{out}}]_{L_*^\perp,L_*^\perp}^{-1}\|=\frac{1}{\sigma_{D-d}([\bSigma_{\mathrm{out}}]_{L_*^\perp,L_*^\perp}^{-1})}\leq \frac{1}{c\epsilon},\label{eq:estimations_magnitude3}\\
&1-\epsilon= \sigma_d([\bSigma_{\mathrm{in}}]_{L_*,L_*})-\sigma_1([\bSigma_{\mathrm{out}}]_{L_*,L_*})\leq \sigma_i([\bSigma]_{L_*,L_*}) \nonumber \\
&\qquad \leq \sigma_1([\bSigma_{\mathrm{in}}]_{L_*,L_*})+\sigma_1([\bSigma_{\mathrm{out}}]_{L_*,L_*})\leq 1/c+\epsilon
\label{eq:estimations_magnitude4}\\
&1-\epsilon-\epsilon/c\leq \sigma_i(\bSigma')\leq 1/c+\epsilon+\epsilon/c,\,\,\text{for all $1\leq i\leq d$}.\label{eq:estimations_magnitude5}
\end{align}
We note that \eqref{eq:estimations_magnitude5} follows from the fact that all eigenvalues of $[\bSigma]_{L_*,L_*}$ are in $[1-\epsilon,1/c+\epsilon]$ (see \eqref{eq:estimations_magnitude4}) and the following inequality, which directly follows from \eqref{eq:estimations_magnitude2} and \eqref{eq:estimations_magnitude3}: 
$\|[\bSigma]_{L_*,L_*^\perp} [\bSigma]_{L_*^\perp,L_*^\perp}^{-1} [\bSigma]_{L_*^\perp,L_*}\|\leq \epsilon/c$.

To prove \eqref{eq:objective_decomposition}, it suffices to show that
\begin{align}
&\frac{\bx^\top \bSigma^{-1} \bx}{(\bU_{L_*}^\top \bx)^\top [\bSigma]_{L_*,L_*}^{-1} (\bU_{L_*}^\top \bx)} = 1 + O(\epsilon), 
\quad \text{for } \bx \in \calX_{\mathrm{in}}, \label{eq:lemma_step2_proof1} \\
&\frac{\bx^\top \bSigma^{-1} \bx}{(\bU_{L_*^\perp}^\top \bx)^\top [\bSigma]_{L_*^\perp,L_*^\perp}^{-1} (\bU_{L_*^\perp}^\top \bx)} = 1 + O(\epsilon), 
\quad \text{for } \bx \in \calX_{\mathrm{out}}, \label{eq:lemma_step2_proof2} \\
&\log \det(\bSigma) = \log \det([\bSigma]_{L_*,L_*}) + \log \det([\bSigma]_{L_*^\perp,L_*^\perp}) + O(\epsilon). \label{eq:lemma_step2_proof3}
\end{align}

\medskip
\noindent
To prove \eqref{eq:lemma_step2_proof1}, observe that for inliers $\bx \in \calX_{\mathrm{in}}$, \eqref{eq:block_representation} implies
\[
\bx^\top \bSigma^{-1} \bx = (\bU_{L_*}^\top \bx)^\top \bSigma' (\bU_{L_*}^\top \bx).
\]
Combining the definition of $\bSigma'$ in \eqref{eq:def_sigma_prime} and the identity $(\bA+\bB)^{-1}-\bA^{-1}=-(\bA+\bB)^{-1}\bB\bA^{-1}$ yields
\[
\bSigma'-[\bSigma]_{L_*,L_*}^{-1} = -(\bSigma')((\bSigma')^{-1}-[\bSigma]_{L_*,L_*})[\bSigma]_{L_*,L_*}^{-1}
= \bSigma'([\bSigma]_{L_*,L_*^\perp}
[\bSigma]_{L_*^\perp,L_*^\perp}^{-1}
[\bSigma]_{L_*^\perp,L_*})[\bSigma]_{L_*,L_*}^{-1}.
\]
Combining the above two equations result in
\begin{equation}
\label{eq:lemma_step2_proof1_1}
\frac{\bx^\top \bSigma^{-1} \bx}{(\bU_{L_*}^\top \bx)^\top [\bSigma]_{L_*,L_*}^{-1} (\bU_{L_*}^\top \bx)}-1 =\frac{(\bU_{L_*}^\top \bx)^\top \bSigma'([\bSigma]_{L_*,L_*^\perp}
[\bSigma]_{L_*^\perp,L_*^\perp}^{-1}
[\bSigma]_{L_*^\perp,L_*})[\bSigma]_{L_*,L_*}^{-1}(\bU_{L_*}^\top \bx)}{(\bU_{L_*}^\top \bx)^\top [\bSigma]_{L_*,L_*}^{-1} (\bU_{L_*}^\top \bx)}.
\end{equation}

Applying the estimations in \eqref{eq:estimations_magnitude2}-\eqref{eq:estimations_magnitude5} results in an estimation of the numerator of the RHS of \eqref{eq:lemma_step2_proof1_1} for $\bx \in \calX_{\mathrm{in}}$:
\begin{multline*}
(\bU_{L_*}^\top \bx)^\top \bSigma' ([\bSigma]_{L_*,L_*^\perp}
[\bSigma]_{L_*^\perp,L_*^\perp}^{-1}
[\bSigma]_{L_*^\perp,L_*})[\bSigma]_{L_*,L_*}^{-1}(\bU_{L_*}^\top \bx) \\
\leq \|\bU_{L_*}^\top \bx\|^2\|[\bSigma]_{L_*^\perp,L_*^\perp}^{-1}\|\|[\bSigma]_{L_*,L_*^\perp}\|^2\frac{\|\bSigma'\|}{\sigma_d([\bSigma]_{L_*,L_*})}\leq O(\epsilon).    
\end{multline*}
We remark that since the dataset is fixed, $\|\bU_{L_*}^\top \bx\|= O(1)$ for $\bx \in \calX_{\mathrm{in}}$. 

Furthermore, applying \eqref{eq:estimations_magnitude4} yields an estimate of the denominator of the RHS of \eqref{eq:lemma_step2_proof1_1} for $\bx \in \calX_{\mathrm{in}}$: 
\begin{equation}\label{eq:inlier_diff}
(\bU_{L_*}^\top \bx)^\top [\bSigma]_{L_*,L_*}^{-1} (\bU_{L_*}^\top \bx)\geq \frac{\|\bU_{L_*}^\top \bx\|^2
}{\sigma_1\!\left([\bSigma]_{L_*,L_*}\right)}=O(1).
\end{equation}
Applying the above two estimates to \eqref{eq:lemma_step2_proof1_1}, concludes the proof of  \eqref{eq:lemma_step2_proof1}.

To prove \eqref{eq:lemma_step2_proof2}, observe that for $\bx \in \calX_{\mathrm{out}}$,  the  denominator of the LHS of \eqref{eq:lemma_step2_proof2} is lower bounded by
\begin{equation}\label{eq:lemma_step2_proof2_1}
(\bU_{L_*^\perp}^\top \bx)^\top [\bSigma]_{L_*^\perp,L_*^\perp}^{-1} (\bU_{L_*^\perp}^\top \bx)\geq \frac{\|\bU_{L_*^\perp}^\top \bx\|^2}{\sigma_1([\bSigma]_{L_*^\perp,L_*^\perp})}=\frac{\|\bU_{L_*^\perp}^\top \bx\|^2}{\sigma_1([\bSigma_{\mathrm{out}}]_{L_*,L_*})}\geq \frac{\|\bU_{L_*^\perp}^\top \bx\|^2}{\sigma_1(\bSigma_{\mathrm{out}})}=O(1/\epsilon).
\end{equation}
We further note that \eqref{eq:block_representation} and \eqref{eq:estimations_magnitude5} implies that 
\begin{equation}\label{eq:inverse_est1}
\|[\bSigma^{-1}]_{L_*,L_*}\|=\|\bSigma'\|=\sigma_1(\bSigma')\leq O(1).
\end{equation}
In addition, \eqref{eq:block_representation},  \eqref{eq:estimations_magnitude2}, \eqref{eq:estimations_magnitude3}, and \eqref{eq:estimations_magnitude5} imply \begin{equation}\label{eq:inverse_est2}\|[\bSigma^{-1}]_{L_*,L_*^\perp}\|\leq \|[\bSigma]_{L_*^\perp,L_*^\perp}^{-1}\|\|[\bSigma]_{L_*^\perp,L_*}\|\sigma_1(\bSigma') \leq O(1),\end{equation} and \begin{equation}\label{eq:inverse_est3}\|[\bSigma]_{L_*^\perp,L_*^\perp}^{-1}
[\bSigma]_{L_*^\perp,L_*}\bSigma'
[\bSigma]_{L_*,L_*^\perp}[\bSigma]_{L_*^\perp,L_*^\perp}^{-1}\|\leq  \|[\bSigma]_{L_*^\perp,L_*^\perp}^{-1}\|^2\|[\bSigma]_{L_*^\perp,L_*}\|^2\sigma_1(\bSigma') \leq O(1).\end{equation}

Using the above estimates we bound the difference of the numerator and denominator of the LHS of \eqref{eq:lemma_step2_proof2} for $\bx \in \calX_{\mathrm{out}}$. We first expand $\bx^\top \bSigma^{-1} \bx$ using the four components of $\bSigma^{-1} $ specified in \eqref{eq:block_representation} and then bound the grouped components using \eqref{eq:inverse_est1}-\eqref{eq:inverse_est3}:
\begin{align}\nonumber
&\bx^\top \bSigma^{-1} \bx - (\bU_{L_*^\perp}\bx)^\top [\bSigma]_{L_*^\perp,L_*^\perp}^{-1}(\bU_{L_*^\perp}\bx)\\
\nonumber
&= \Big((\bU_{L_*}\bx)^\top [\bSigma^{-1}]_{L_*,L_*}(\bU_{L_*}\bx)+(\bU_{L_*^\perp}\bx)^\top [\bSigma^{-1}]_{L_*^\perp,L_*}(\bU_{L_*}\bx)\\\nonumber&+(\bU_{L_*}\bx)^\top [\bSigma^{-1}]_{L_*,L_*^\perp}(\bU_{L_*^\perp}\bx)+(\bU_{L_*^\perp}\bx)^\top [\bSigma^{-1}]_{L_*^\perp,L_*^\perp}(\bU_{L_*^\perp}\bx)\Big)-(\bU_{L_*^\perp}\bx)^\top [\bSigma]_{L_*^\perp,L_*^\perp}^{-1}(\bU_{L_*^\perp}\bx)\\\nonumber
&=(\bU_{L_*}\bx)^\top [\bSigma^{-1}]_{L_*,L_*}(\bU_{L_*}\bx)+(\bU_{L_*^\perp}\bx)^\top [\bSigma^{-1}]_{L_*^\perp,L_*}(\bU_{L_*}\bx)+(\bU_{L_*}\bx)^\top [\bSigma^{-1}]_{L_*,L_*^\perp}(\bU_{L_*^\perp}\bx)\\\nonumber&+(\bU_{L_*^\perp}\bx)^\top [\bSigma]_{L_*^\perp,L_*^\perp}^{-1}
[\bSigma]_{L_*^\perp,L_*}\bSigma'
[\bSigma]_{L_*,L_*^\perp}[\bSigma]_{L_*^\perp,L_*^\perp}^{-1}(\bU_{L_*^\perp}\bx)\\\nonumber&\leq\|\bU_{L_*}\bx\|^2\|[\bSigma^{-1}]_{L_*,L_*}\|+2\|\bU_{L_*}\bx\|\|\bU_{L_*^\perp}\bx\|\|[\bSigma^{-1}]_{L_*^\perp,L_*}\|\\\nonumber&+ \|\bU_{L_*^\perp}\bx\|^2\|[\bSigma]_{L_*^\perp,L_*^\perp}^{-1}
[\bSigma]_{L_*^\perp,L_*}\bSigma'
[\bSigma]_{L_*,L_*^\perp}[\bSigma]_{L_*^\perp,L_*^\perp}^{-1}\|
\\&\leq  O(1).\label{eq:outlier_diff}
\end{align}
Combining this bound with \eqref{eq:lemma_step2_proof2_1} concludes the proof of \eqref{eq:lemma_step2_proof2}.

To prove \eqref{eq:lemma_step2_proof3}, we first note that since $\bSigma = \bSigma_{\mathrm{in}}+ \bSigma_{\mathrm{out}}$ and $\bSigma_{\mathrm{in}}=\begin{pmatrix}[\bSigma_{\mathrm{in}}]_{L_*,L_*}&0\\0&0\end{pmatrix}$,
\begin{multline*}
\left\|\bSigma-\begin{pmatrix}[\bSigma]_{L_*,L_*}&0\\0&0\end{pmatrix}\right\|=\left\|\bSigma_{\mathrm{out}}-\begin{pmatrix}[\bSigma_{\mathrm{out}}]_{L_*,L_*}&0\\0&0\end{pmatrix}\right\|\\=  \|\bSigma_{\mathrm{out}}-\bP_{L_*}\bSigma_{\mathrm{out}}\bP_{L_*}\| \leq 2\|\bSigma_{\mathrm{out}}\|=2\sigma_1(\bSigma_{\mathrm{out}}).    
\end{multline*}
Combining this bound with Weyl's inequality implies that for $1\leq i\leq d$, 
\[
\left|\sigma_i([\bSigma]_{L_*,L_*})-\sigma_i(\bSigma)\right|=\left|\sigma_i\begin{pmatrix}[\bSigma]_{L_*,L_*}&0\\0&0\end{pmatrix}-\sigma_i(\bSigma)\right|\leq \left\|\bSigma-\begin{pmatrix}[\bSigma]_{L_*,L_*}&0\\0&0\end{pmatrix}\right\|\leq  2\sigma_1([\bSigma]_{\mathrm{out}}).
\]
 Since $\sigma_i([\bSigma]_{L_*,L_*})=O(1)$ (see \eqref{eq:estimations_magnitude4}) and $\sigma_1([\bSigma]_{\mathrm{out}})\leq \epsilon$ (see \eqref{eq:estimations_magnitude1}) , we have
\begin{equation}\label{eq:lemma_step2_proof3.1}
\left|\frac{\sigma_i([\bSigma]_{L_*,L_*})}{\sigma_i(\bSigma)}-1\right|=\left|\frac{\sigma_i([\bSigma]_{L_*,L_*})-\sigma_i(\bSigma)}{\sigma_i(\bSigma)}\right|\leq O(\epsilon),\,\,1\leq i\leq d.
\end{equation}

For the other indices, that is, $d+1\leq i\leq D$, note that \eqref{eq:block_representation} implies 
\[
\bSigma^{-1}=\begin{pmatrix}
0 & 0\\
0& [\bSigma]_{L_*^\perp,L_*^\perp}^{-1}
\end{pmatrix}+\begin{pmatrix}
\bSigma' & -\bSigma' [\bSigma]_{L_*,L_*^\perp}[\bSigma]_{L_*^\perp,L_*^\perp}^{-1}\\
-[\bSigma]_{L_*^\perp,L_*^\perp}^{-1}[\bSigma]_{L_*^\perp,L_*}\bSigma'& [\bSigma]_{L_*^\perp,L_*^\perp}^{-1}[\bSigma]_{L_*^\perp,L_*}\bSigma'[\bSigma]_{L_*^\perp,L_*^\perp}^{-1}[\bSigma]_{L_*^\perp,L_*} 
\end{pmatrix}.
\]
To apply Weyl's inequality, we observe $$\sigma_{D-i}(\bSigma^{-1})=\sigma_i(\bSigma)^{-1} \ \text{ and } \ \sigma_{D-i}\begin{pmatrix}
0 & 0\\
0& [\bSigma]_{L_*^\perp,L_*^\perp}^{-1}
\end{pmatrix}=\sigma_{D-i}([\bSigma]_{L_*^\perp,L_*^\perp}^{-1})=\sigma_{i-d}([\bSigma]_{L_*^\perp,L_*^\perp})^{-1},$$ where we note that $[\bSigma]_{L_*^\perp,L_*^\perp}$ has size $D-d$. 
Applying the above two equations, \eqref{eq:inverse_est1}-\eqref{eq:inverse_est3}, and Weyl's inequality yields for $ d+1\leq i\leq D$, 
\begin{align*}
&\left|\sigma_i(\bSigma)^{-1}- \sigma_{i-d}([\bSigma]_{L_*^\perp,L_*^\perp})^{-1}\right| 
= \left|\sigma_{D-i}(\bSigma^{-1})- \sigma_{D-i}\!\left(\begin{pmatrix}
0 & 0\\
0& [\bSigma]_{L_*^\perp,L_*^\perp}^{-1}
\end{pmatrix}\right)\right|\\
&\leq \left\|\begin{pmatrix}
\bSigma' & -\bSigma' [\bSigma]_{L_*,L_*^\perp}[\bSigma]_{L_*^\perp,L_*^\perp}^{-1}\\
-[\bSigma]_{L_*^\perp,L_*^\perp}^{-1}[\bSigma]_{L_*^\perp,L_*}\bSigma'& [\bSigma]_{L_*^\perp,L_*^\perp}^{-1}[\bSigma]_{L_*^\perp,L_*}\bSigma'[\bSigma]_{L_*^\perp,L_*^\perp}^{-1}[\bSigma]_{L_*^\perp,L_*} 
\end{pmatrix}\right\| \leq O(1).
\end{align*}
Note that \eqref{eq:inverse_est1}-\eqref{eq:inverse_est3} imply that the spectral norm of each one of the four blocks of the above matrix is $O(1)$. Since $D$ is fixed, their Frobenius norm is also $O(1)$. Therefore the Frobenius norm, and consequently the spectral norm, of the block matrix, is also $O(1)$. 

In addition, from \eqref{eq:estimations_magnitude1},
\[
c\epsilon\leq \sigma_{D-d}\!\left(
[\bSigma_{\mathrm{out}}]_{L_*^\perp,L_*^\perp}
\right)= \sigma_{D-d}([\bSigma]_{L_*^\perp,L_*^\perp})\leq \sigma_1([\bSigma]_{L_*^\perp,L_*^\perp})\leq \sigma_1(\bSigma_{\mathrm{out}})\leq \epsilon,
\] so for all $d+1\leq i\leq D$, $\sigma_{i-d}([\bSigma]_{L_*^\perp,L_*^\perp})^{-1}=O(1/\epsilon)$. Therefore, the above estimations imply that
\begin{equation}\label{eq:lemma_step2_proof3.2}
\left|\frac{\sigma_{i-d}([\bSigma]_{L_*^\perp,L_*^\perp})}{\sigma_i(\bSigma)}-1\right| = \left|\frac{\sigma_i(\bSigma)^{-1}-\sigma_{i-d}([\bSigma]_{L_*^\perp,L_*^\perp})^{-1}}{\sigma_{i-d}([\bSigma]_{L_*^\perp,L_*^\perp})^{-1}}\right|  \leq O(\epsilon),\,\,\text{for all $d+1\leq i\leq D$}.
\end{equation}

Combining \eqref{eq:lemma_step2_proof3.1} and \eqref{eq:lemma_step2_proof3.2} with the fact that the determinant of a matrix equals the product of its eigenvalues 
\begin{multline*}
 \log \det([\bSigma]_{L_*,L_*}) + \log \det([\bSigma]_{L_*^\perp,L_*^\perp})-\log \det(\bSigma)  \\=\log\left(\prod_{i=1}^d\frac{\sigma_i([\bSigma]_{L_*,L_*})}{\sigma_i(\bSigma)}\prod_{i=d+1}^D\frac{\sigma_{i-d}([\bSigma]_{L_*^\perp,L_*^\perp})}{\sigma_i(\bSigma)}\right) = \log(1+O(\epsilon))= O(\epsilon).   
\end{multline*}
This clearly concludes the proof of \eqref{eq:lemma_step2_proof3}, and consequently, the proof of  \eqref{eq:objective_decomposition}.

To prove \eqref{eq:operator_decomposition}, we first show that both terms in its LHS, that is,  ${
[T_{\calX}(\bSigma)]_{L_*,L_*}
}/{
\tr\!\left(
[T_{\calX}(\bSigma)]_{L_*,L_*}
\right)
}$ and $T_{\tilde{\calX}_{\mathrm{in}}}
\!\left(
[\bSigma]_{L_*,L_*}
\right)$, are weighted sums of $\tilde{\bx}\tilde{\bx}^\top$, where $\tilde{\bx}=\bU_{L_*}^\top\bx\in\R^d$, with weights 
$$w_{\mathrm{in},1}(\bx)=\frac{1}{(\bU_{L_*}^\top \bx)^\top [\bSigma]_{L_*,L_*}^{-1} (\bU_{L_*}^\top \bx)} \ \text{ and } \  w_{\mathrm{in},2}(\bx)=\frac{1}{\bx^\top \bSigma^{-1} \bx },$$ 
respectively, and that both terms have trace $1$. Indeed, 
\begin{equation}\label{eq:inlier_operator1}
T_{\tilde{\calX}_{\mathrm{in}}}
\!\left(
[\bSigma]_{L_*,L_*}
\right)
=
\frac{\sum_{\tilde{\bx}\in\tilde{\calX}_{\mathrm{in}}}\frac{\tilde{\bx} \tilde{\bx}^\top}{\tilde{\bx}^\top([\bSigma]_{L_*,L_*})^{-1}\tilde{\bx}}}{\tr\left(\sum_{\tilde{\bx}\in\tilde{\calX}_{\mathrm{in}}}\frac{\tilde{\bx} \tilde{\bx}^\top}{\tilde{\bx}^\top([\bSigma]_{L_*,L_*})^{-1}\tilde{\bx}}\right)}
=
\frac{\sum_{\tilde{\bx}\in\tilde{\calX}_{\mathrm{in}}}w_{\mathrm{in},1}(\tilde{\bx})\tilde{\bx}\tilde{\bx}^\top}{\tr\left(\sum_{\tilde{\bx}\in\tilde{\calX}_{\mathrm{in}}}w_{\mathrm{in},1}(\tilde{\bx})\tilde{\bx}\tilde{\bx}^\top\right)},\end{equation}
\begin{multline}
\label{eq:inlier_operator2} \frac{
[T_{\calX}(\bSigma)]_{L_*,L_*}
}{
\tr\!\left(
[T_{\calX}(\bSigma)]_{L_*,L_*}
\right)
}=\frac{\left[\sum_{\bx\in\calX}\frac{\bx \bx^\top}{\bx^\top(\bSigma)^{-1}\bx}\right]_{L_*,L_*}}{\tr\left(\left[\sum_{\bx\in\calX}\frac{\bx \bx^\top}{\bx^\top(\bSigma)^{-1}\bx}\right]_{L_*,L_*}\right)}
\\=\frac{\sum_{\tilde{\bx}\in\tilde{\calX}_{\mathrm{in}}}w_{\mathrm{in},2}(\tilde{\bx})\tilde{\bx}\tilde{\bx}^\top+\left[\sum_{\bx\in\calX_{\mathrm{out}}}\frac{\bx \bx^\top}{\bx^\top(\bSigma)^{-1}\bx}\right]_{L_*,L_*}}{\tr\left(\sum_{\tilde{\bx}\in\tilde{\calX}_{\mathrm{in}}}w_{\mathrm{in},2}(\tilde{\bx})\tilde{\bx}\tilde{\bx}^\top\right)+\tr\left(\left[\sum_{\bx\in\calX_{\mathrm{out}}}\frac{\bx \bx^\top}{\bx^\top(\bSigma)^{-1}\bx}\right]_{L_*,L_*}\right)}.   
\end{multline}

We first show that the LHS of \eqref{eq:inlier_operator2} is approximated by $
{\sum_{\tilde{\bx}\in\tilde{\calX}_{\mathrm{in}}}w_{\mathrm{in},2}(\tilde{\bx})\tilde{\bx}\tilde{\bx}^\top}/{\tr\left(\sum_{\tilde{\bx}\in\tilde{\calX}_{\mathrm{in}}}w_{\mathrm{in},2}(\tilde{\bx})\tilde{\bx}\tilde{\bx}^\top\right)}$. 
We note that \eqref{eq:lemma_step2_proof2_1}  and \eqref{eq:outlier_diff} imply that 
\[
\bx^\top \bSigma^{-1} \bx \geq  O(1/\epsilon) \quad \text{for } \bx \in \calX_{\mathrm{out}}.\] 
In addition, a similar argument to \eqref{eq:inlier_diff}, using \eqref{eq:estimations_magnitude4} and  $\|\bU_{L_*}^\top \bx\|= O(1)$ for $\bx \in \calX_{\mathrm{in}}$, implies 
\begin{equation}\label{eq:inlier_diff2}
\frac{1}{w_{\mathrm{in},1}(\bx)}=(\bU_{L_*}^\top \bx)^\top [\bSigma]_{L_*,L_*}^{-1} (\bU_{L_*}^\top \bx)\leq \frac{\|\bU_{L_*}^\top \bx\|^2
}{\sigma_d([\bSigma]_{L_*,L_*})}=O(1) \ \text{ for } \bx \in \calX_{\mathrm{in}}.
\end{equation}
Therefore, in the denominator of the RHS of \eqref{eq:inlier_operator2}, $\tr\left(\left[\sum_{\bx\in\calX_{\mathrm{out}}}\frac{\bx \bx^\top}{\bx^\top(\bSigma)^{-1}\bx}\right]_{L_*,L_*}\right)$ is smaller than the order of $O(\epsilon)$ while the other component of $\tr\left(\sum_{\tilde{\bx}\in\tilde{\calX}_{\mathrm{in}}}w_{\mathrm{in},2}(\tilde{\bx})\tilde{\bx}\tilde{\bx}^\top\right)$ is larger than the order of $O(1)$. That is, 
\[
\beta:=\frac{\tr\left(\left[\sum_{\bx\in\calX_{\mathrm{out}}}\frac{\bx \bx^\top}{\bx^\top(\bSigma)^{-1}\bx}\right]_{L_*,L_*}\right)}{\tr\left(\sum_{\tilde{\bx}\in\tilde{\calX}_{\mathrm{in}}}w_{\mathrm{in},2}(\tilde{\bx})\tilde{\bx}\tilde{\bx}^\top\right)}=O(\epsilon)
\]
and
\begin{align}\nonumber
&\left\|\frac{\sum_{\tilde{\bx}\in\tilde{\calX}_{\mathrm{in}}}w_{\mathrm{in},2}(\tilde{\bx})\tilde{\bx}\tilde{\bx}^\top}{\tr\left(\sum_{\tilde{\bx}\in\tilde{\calX}_{\mathrm{in}}}w_{\mathrm{in},2}(\tilde{\bx})\tilde{\bx}\tilde{\bx}^\top\right)+\tr\left(\left[\sum_{\bx\in\calX_{\mathrm{out}}}\frac{\bx \bx^\top}{\bx^\top(\bSigma)^{-1}\bx}\right]_{L_*,L_*}\right)}-\frac{\sum_{\tilde{\bx}\in\tilde{\calX}_{\mathrm{in}}}w_{\mathrm{in},2}(\tilde{\bx})\tilde{\bx}\tilde{\bx}^\top}{\tr\left(\sum_{\tilde{\bx}\in\tilde{\calX}_{\mathrm{in}}}w_{\mathrm{in},2}(\tilde{\bx})\tilde{\bx}\tilde{\bx}^\top\right)}\right\|_F\\
&=\left\|\frac{\sum_{\tilde{\bx}\in\tilde{\calX}_{\mathrm{in}}}w_{\mathrm{in},2}(\tilde{\bx})\tilde{\bx}\tilde{\bx}^\top}{\tr\left(\sum_{\tilde{\bx}\in\tilde{\calX}_{\mathrm{in}}}w_{\mathrm{in},2}(\tilde{\bx})\tilde{\bx}\tilde{\bx}^\top\right)}\right\|_F\frac{\beta}{1+\beta}=O(\epsilon),\label{eq:operator_decomposition_proof2}
\end{align}
where the last step applies the fact that since $\frac{\sum_{\tilde{\bx}\in\tilde{\calX}_{\mathrm{in}}}w_{\mathrm{in},2}(\tilde{\bx})\tilde{\bx}\tilde{\bx}^\top}{\tr\left(\sum_{\tilde{\bx}\in\tilde{\calX}_{\mathrm{in}}}w_{\mathrm{in},2}(\tilde{\bx})\tilde{\bx}\tilde{\bx}^\top\right)}$ is in $S_{++}$ and has trace 1, its Frobenius norm is $O(1)$. Similarly, 
\begin{align}\nonumber
&\left\|\frac{\left[\sum_{\bx\in\calX_{\mathrm{out}}}\frac{\bx \bx^\top}{\bx^\top(\bSigma)^{-1}\bx}\right]_{L_*,L_*}}{\tr\left(\sum_{\tilde{\bx}\in\tilde{\calX}_{\mathrm{in}}}w_{\mathrm{in},2}(\tilde{\bx})\tilde{\bx}\tilde{\bx}^\top\right)+\tr\left(\left[\sum_{\bx\in\calX_{\mathrm{out}}}\frac{\bx \bx^\top}{\bx^\top(\bSigma)^{-1}\bx}\right]_{L_*,L_*}\right)}\right\|_F\\=&\frac{\beta}{\beta+1}\left\|\frac{\left[\sum_{\bx\in\calX_{\mathrm{out}}}\frac{\bx \bx^\top}{\bx^\top(\bSigma)^{-1}\bx}\right]_{L_*,L_*}}{\tr\left(\left[\sum_{\bx\in\calX_{\mathrm{out}}}\frac{\bx \bx^\top}{\bx^\top(\bSigma)^{-1}\bx}\right]_{L_*,L_*}\right)}\right\|_F\leq O(\epsilon).
\label{eq:operator_decomposition_proof3}
\end{align}

We are now ready to conclude \eqref{eq:operator_decomposition}. Based on \eqref{eq:lemma_step2_proof1}, we have that the ratio between $w_{\mathrm{in},1}(\tilde{\bx})$ and $w_{\mathrm{in},2}(\tilde{\bx})$ is  $1+O(\epsilon)$. As a result, the ratio between $\frac{w_{\mathrm{in},1}(\tilde{\bx})}{\tr\left(\sum_{\tilde{\bx}\in\tilde{\calX}_{\mathrm{in}}}w_{\mathrm{in},1}(\tilde{\bx})\tilde{\bx}\tilde{\bx}^\top\right)}$ and $\frac{w_{\mathrm{in},2}(\tilde{\bx})}{\tr\left(\sum_{\tilde{\bx}\in\tilde{\calX}_{\mathrm{in}}}w_{\mathrm{in},2}(\tilde{\bx})\tilde{\bx}\tilde{\bx}^\top\right)}$ is also on the order of $1+O(\epsilon)$. Since both values are bounded above by $1/\|\tilde{\bx}\|^2=O(1)$, we have 
\[
\left|\frac{w_{\mathrm{in},1}(\tilde{\bx})}{\tr\left(\sum_{\tilde{\bx}\in\tilde{\calX}_{\mathrm{in}}}w_{\mathrm{in},1}(\tilde{\bx})\tilde{\bx}\tilde{\bx}^\top\right)}-\frac{w_{\mathrm{in},2}(\tilde{\bx})}{\tr\left(\sum_{\tilde{\bx}\in\tilde{\calX}_{\mathrm{in}}}w_{\mathrm{in},2}(\tilde{\bx})\tilde{\bx}\tilde{\bx}^\top\right)}\right|\leq O(\epsilon)
\]
and
\begin{equation}\label{eq:operator_decomposition_proof4}
\left\|\frac{\sum_{\tilde{\bx}\in\tilde{\calX}_{\mathrm{in}}}w_{\mathrm{in},1}(\tilde{\bx})\tilde{\bx}\tilde{\bx}^\top}{\tr\left(\sum_{\tilde{\bx}\in\tilde{\calX}_{\mathrm{in}}}w_{\mathrm{in},1}(\tilde{\bx})\tilde{\bx}\tilde{\bx}^\top\right)}-\frac{\sum_{\tilde{\bx}\in\tilde{\calX}_{\mathrm{in}}}w_{\mathrm{in},2}(\tilde{\bx})\tilde{\bx}\tilde{\bx}^\top}{\tr\left(\sum_{\tilde{\bx}\in\tilde{\calX}_{\mathrm{in}}}w_{\mathrm{in},2}(\tilde{\bx})\tilde{\bx}\tilde{\bx}^\top\right)}\right\|_F\leq O(\epsilon).
\end{equation}
Combining the last equation with \eqref{eq:inlier_operator1}, \eqref{eq:inlier_operator2}, \eqref{eq:operator_decomposition_proof2}, and \eqref{eq:operator_decomposition_proof3}, 
\eqref{eq:operator_decomposition} is proved.

Finally, \eqref{eq:operator_decomposition2} follows from \eqref{eq:lemma_step2_proof2}. Indeed, both $T_{\widehat{\calX}_{\mathrm{out}}}
\!\left(
[\bSigma]_{L_*^\perp,L_*^\perp}
\right)$ and  $
[T_{\calX}(\bSigma)]_{L_*^\perp,L_*^\perp}
$ are weighted sums of $(\bU_{L_*^\perp}^\top\bx)(\bU_{L_*^\perp}^\top\bx)^\top$, where the first has weights $w_{\mathrm{out},1}(\bx)={1}/{(\bU_{L_*^\perp}^\top \bx)^\top [\bSigma]_{L_*^\perp,L_*^\perp}^{-1} (\bU_{L_*^\perp}^\top \bx)}$ and the second has weights  $w_{\mathrm{out},2}(\bx)= {1}/{\bx^\top\bSigma^{-1}\bx}$  (note that both weights will be normalized so that the trace is $1$). That is,
\begin{align*}
T_{\widehat{\calX}_{\mathrm{out}}}
\!\left(
[\bSigma]_{L_*^\perp,L_*^\perp}
\right)=\frac{\sum_{{\bx}\in{\calX}_{\mathrm{out}}}w_{\mathrm{out},1}(\bx)(\bU_{L_*^\perp}^\top\bx)(\bU_{L_*^\perp}^\top\bx)^\top}{\tr\left(\sum_{{\bx}\in{\calX}_{\mathrm{out}}}w_{\mathrm{out},1}(\bx)(\bU_{L_*^\perp}^\top\bx)(\bU_{L_*^\perp}^\top\bx)^\top\right)},\\
[T_{\calX}(\bSigma)]_{L_*^\perp,L_*^\perp}=\frac{\sum_{{\bx}\in{\calX}_{\mathrm{out}}}w_{\mathrm{out},2}(\bx)(\bU_{L_*^\perp}^\top\bx)(\bU_{L_*^\perp}^\top\bx)^\top}{\tr\left(\sum_{{\bx}\in{\calX}_{\mathrm{out}}}w_{\mathrm{out},2}(\bx)(\bU_{L_*^\perp}^\top\bx)(\bU_{L_*^\perp}^\top\bx)^\top\right)}.
\end{align*}
By \eqref{eq:lemma_step2_proof2}, the ratio of the two weights $w_{\mathrm{out},1}(\bx)/w_{\mathrm{out},2}(\bx)$ is on the order of $1+O(\epsilon)$, so using the same argument as in the proof of  \eqref{eq:operator_decomposition_proof4}, $\|
[T_{\calX}(\bSigma)]_{L_*^\perp,L_*^\perp}
-T_{\widehat{\calX}_{\mathrm{out}}}
\!\left(
[\bSigma]_{L_*^\perp,L_*^\perp}
\right)\|_F$ is also on the order of $O(\epsilon)$ and \eqref{eq:operator_decomposition2} is proved.

\qed


\section{Proof of Theorem \ref{thm:zhang16_extend} for the Case of $\boldsymbol{\dssnr>1}$}
\label{sec:proof_thm_zhang16_extend_part2}

We first prove the following lemma, which shows that there exists a linear transformation represented by $\bT \in \R^{D \times D}$ such that the transformed dataset $\calX$ satisfies certain properties for both the inliers and the outliers.

\begin{lemma}\label{lemma:largest_dssnr}
Let $\mathcal{X} \subset \R^D$ be a dataset with underlying subspace $L_*$ such that $\dssnr > 1$, and suppose that $L_*$ is the unique subspace satisfying \eqref{eq:assumption3_thm}. Then there exists $\mathbf{T} \in \R^{D \times D}$ such that the objects obtained from its application to $\bx$, $L_*$ and $\calX$, which we denote by $\hat{\bx}$, $\hat{L}_*$ and $\hat{\calX}$, respectively, satisfy
\begin{align}\label{eq:equality_in3}
&
\frac{D}{N}\sum_{\hat{\bx}\in \hat{\calX}_{\mathrm{in}}}
\frac{(\bU_{\hat{L}_*}^\top \hat{\bx})(\bU_{\hat{L}_*}^\top \hat{\bx})^\top}
{\|\bU_{\hat{L}_*}^\top \hat{\bx}\|^2}
= \dssnr_{\mathrm{in}}\!\cdot  \bI, \ \text{ where }  \dssnr_{\mathrm{in}}:=\frac{D}{N}\frac{|\calX_1|}{\dim(L_1)},
\\
& \dssnr_{\mathrm{out}}:=\frac{D}{N}\left\|
\sum_{\hat{\bx}\in \hat{\calX}_{\mathrm{out}}}
\frac{(\bU_{\hat{L}_*^\perp}^\top \hat{\bx})(\bU_{\hat{L}_*^\perp}^\top \hat{\bx})^\top}
{\|\bU_{\hat{L}_*^\perp}^\top \hat{\bx}\|^2}
\right\|
< \dssnr_{\mathrm{in}}.\label{eq:equality_out3}
\end{align}

\end{lemma}

Since the TME estimator is invariant under linear transformations and rescaling of individual data points, WLOG we may replace each $\bx$ by $\hat{\bx}=\bT\bx$ and then scale each inlier by $\| \hat{\bx}\|$ and each  outlier by $\|\bU_{\hat{L}_*^\perp}^\top\hat{\bx}\|$.  Hence, we may assume that
\begin{align}\label{eq:equality_in2}
&\text{if $\bx \in \calX_{\mathrm{in}}$, then $\bx \in L_*$, $\|\bx\|=1$, and }  \frac{D}{N}\sum_{\bx\in \calX_{\mathrm{in}}} (\bU_{L_*}^\top\bx)(\bU_{L_*}^\top\bx)^\top
= \dssnr_{\mathrm{in}}\cdot\,\bI,\\
&\text{if $\bx \in \calX_{\mathrm{out}}$, then $\|\bU_{L_*^\perp}^\top\bx\|=1$, and } \frac{D}{N}\left\|
\sum_{\bx\in \calX_{\mathrm{out}}}
(\bU_{L_*^\perp}^\top\bx)(\bU_{L_*^\perp}^\top\bx)^\top
\right\|
= \dssnr_{\mathrm{out}}.
\label{eq:equality_out2}
\end{align}


We first show that the ratio between the outlier weights and the inlier weights
\[
r_k := \frac{\max_{\bx_i\in \calX_{\mathrm{out}}} w_i^{(k)}}{\min_{\bx_i\in \calX_{\mathrm{in}}} w_i^{(k)}}
\]
converges to zero. 
Applying \eqref{eq:w_update}, basic bounds, and finally \eqref{eq:equality_in2}, we note that for any $\bx_i \in \calX_{\mathrm{in}}$, 
\begin{align}\label{eq:weight_inlier}
w_i^{(k+1)}=&\frac{D}{N\,\bx_i^\top \left(\sum_{\bx_j\in\calX} w_j^{(k)} \bx_j \bx_j^\top \right)^{-1}\bx_i}
\geq
\frac{D}{N\,\bx_i^\top \left(\sum_{\bx_j\in\calX_{\mathrm{in}}} w_j^{(k)} \bx_j \bx_j^\top \right)^{-1}\bx_i}
\\\geq&\nonumber
\frac{D\,\min_{\bx_j\in \calX_{\mathrm{in}}} w_j^{(k)}}{N\,\bx_i^\top \left(\sum_{\bx_j\in\calX_{\mathrm{in}}} \bx_j \bx_j^\top \right)^{-1}\bx_i}
=
\dssnr_{\mathrm{in}}\,\min_{\bx_j\in \calX_{\mathrm{in}}} w_j^{(k)}.
\end{align}

Similarly, for any $\bx_i \in \calX_{\mathrm{out}}$, define $[\bx]_{L_*^\perp}=\bU_{L_*^\perp}^\top\bx$ and apply a similar analysis to obtain 
\begin{align}\nonumber
w_i^{(k+1)}
\leq&
\frac{D}{N\,\bx_i^\top \lim_{\lambda\rightarrow\infty}\left(\lambda\cdot \bP_{L_*}+\sum_{\bx_j\in\calX_{\mathrm{out}}} w_j^{(k)} \bx_j \bx_j^\top \right)^{-1}\bx_i}\\
=&\frac{D}{N\,[\bx_i]_{L_*^\perp}^\top \left[\sum_{\bx_j\in\calX_{\mathrm{out}}} w_j^{(k)} \bx_j \bx_j^\top\right]_{L_*^\perp,L_*^\perp}^{-1}[\bx_i]_{L_*^\perp}}
\label{eq:weight_outlier}
\\
\leq&
\frac{D\max_{\bx_j\in \calX_{\mathrm{out}}} w_j^{(k)}}{N\,[\bx_i]_{L_*^\perp}^\top \left([\sum_{\bx_j\in\calX_{\mathrm{out}}} \bx_j \bx_j^\top]_{L_*^\perp,L_*^\perp} \right)^{-1}[\bx_i]_{L_*^\perp}}
\leq
\dssnr_{\mathrm{out}}\,\max_{\bx_j\in \calX_{\mathrm{out}}} w_j^{(k)}.
\nonumber
\end{align}

Since $\dssnr > \dssnr_{\mathrm{out}}$, it follows from \eqref{eq:weight_inlier} and \eqref{eq:weight_outlier} that
\begin{equation}\label{eq:ratio_convergence}
r_{k+1} \leq \frac{\dssnr_{\mathrm{out}}}{\dssnr_{\mathrm{in}}}\, r_k.
\end{equation}
Therefore, $r_k$ converges to $0$ at a linear rate.

Next, we show that $\min_{\bx_i\in \calX_{\mathrm{in}}} w_i^{(k)}$ does not converge to $0$. 
We first note that for any $\bx_i \in \calX_{\mathrm{in}}$
\begin{equation*}
w_i^{(k+1)}
=
\frac{D}{N\,\bx_i^\top \left(\sum_{\bx_j\in\calX} w_j^{(k)} \bx_j \bx_j^\top\right)^{-1}\bx_i} =
\frac{D}{N\,[\bx_i]_{L_*}^\top
\left[\left(\sum_{\bx_j\in\calX} w_j^{(k)} \bx_j \bx_j^\top\right)^{-1}\right]_{L_*,L_*}
[\bx_i]_{L_*}}. 
\end{equation*}
To bound the RHS of the above equality, 
we recall that the $(L_*, L_*)$ block of $\mathbf{X}^{-1}$ can be expressed as a Schur complement (see, e.g., \cite[Section 9.1.5]{IMM2012-03274}). Specifically, if $\mathbf{X} \succeq 0$, then
\[[\bX^{-1}]_{L_*,L_*}
=
\left([\bX]_{L_*,L_*}-[\bX]_{L_*,L_*^\perp}[\bX]_{L_*^\perp,L_*^\perp}^{-1}[\bX]_{L_*^\perp,L_*}\right)^{-1}
\psdgeq [\bX]_{L_*,L_*}^{-1}.\] 
Combining the above two equations yields 
\begin{align*}
w_i^{(k+1)}
\le \frac{D}{N\,[\bx_i]_{L_*}^\top
\left(\left[\sum_{\bx_j\in\calX} w_j^{(k)} \bx_j \bx_j^\top\right]_{L_*,L_*}\right)^{-1}
[\bx_i]_{L_*}}.
\end{align*}
To further bound the RHS of the above equation, we apply 
\[\max_{\bx_j\in \calX_{\mathrm{out}}} w_j^{(k)}=r_k\min_{\bx_j\in \calX_{\mathrm{in}}} w_j^{(k)}\leq r_k\max_{\bx_j\in \calX_{\mathrm{in}}} w_j^{(k)}.\]
Therefore, combining the above two equations, then applying \eqref{eq:equality_in2}, and next applying basic estimates using the following constant:
$$C_0
:=\frac{\left\|[\sum_{\bx_j\in\calX_{\mathrm{out}}} \bx_j \bx_j^\top]_{L_*,L_*}\right\|}{\dssnr_{\mathrm{in}} \!\cdot  N/D},$$
yields
\begin{align*}
w_i^{(k+1)}
&\le 
\frac{D\,\max_{\bx_j\in \calX_{\mathrm{in}}} w_j^{(k)}}
{N\,[\bx_i]_{L_*}^\top
\left(
\left[\sum_{\bx_j\in\calX_{\mathrm{in}}} \bx_j \bx_j^\top
+ r_k \sum_{\bx_j\in\calX_{\mathrm{out}}} \bx_j \bx_j^\top\right]_{L_*,L_*}
\right)^{-1}
[\bx_i]_{L_*}} \\
&=\frac{D\,\max_{\bx_j\in \calX_{\mathrm{in}}} w_j^{(k)}}
{N\,[\bx_i]_{L_*}^\top
\left(\dssnr_{\mathrm{in}} \frac{N}{D}\bI
+ \left[r_k \sum_{\bx_j\in\calX_{\mathrm{out}}} \bx_j \bx_j^\top\right]_{L_*,L_*}\right)^{-1}
[\bx_i]_{L_*}} \\
&\leq \frac{D\,\max_{\bx_j\in \calX_{\mathrm{in}}} w_j^{(k)}}
{N\,[\bx_i]_{L_*}^\top
\left(\dssnr_{\mathrm{in}} \frac{N}{D}\bI
+ C_0 r_k \dssnr_{\mathrm{in}} \frac{N}{D}\bI\right)^{-1}
[\bx_i]_{L_*}} \\
&=
\dssnr_{\mathrm{in}}(1+r_k C_0)\,\max_{\bx_j\in \calX_{\mathrm{in}}} w_j^{(k)}.
\end{align*}
Combining the above equation with \eqref{eq:weight_inlier}, results in
\begin{equation}\label{eq:ratio_convergence2}
\frac{\max_{\bx_i\in \calX_{\mathrm{in}}} w_i^{(k+1)} \big/ \min_{\bx_i\in \calX_{\mathrm{in}}} w_i^{(k+1)}}
{\max_{\bx_i\in \calX_{\mathrm{in}}} w_i^{(k)} \big/ \min_{\bx_i\in \calX_{\mathrm{in}}} w_i^{(k)}}
\le 1+r_k C_0.
\end{equation}

By \eqref{eq:ratio_convergence}, the sequence \(r_k\) is dominated by a convergent geometric sequence. It follows that the infinite product
$\prod_{k=1}^{\infty} (1+r_kC_0 )$
converges. Indeed, for all sufficiently large \(k\), we have \(0 \le C_0 r_k \le \tfrac12\), and hence
$\log(1+r_kC_0) \le 2 C_0 r_k$. Since $\sum_kr_k$ converges, $\sum_{k=1}^{\infty} \log(1+r_kC_0)$
converges as well, which implies
\[
\prod_{k=1}^{\infty} (1+r_kC_0)
= \exp\!\left(\sum_{k=1}^{\infty} \log(1+r_kC_0)\right)
\]
converges as well.
Therefore, \eqref{eq:ratio_convergence2} implies  
\[
\frac{\max_{\bx_i\in \calX_{\mathrm{in}}} w_i^{(K)}}{\min_{\bx_i\in \calX_{\mathrm{in}}} w_i^{(K)}}\leq \frac{\max_{\bx_i\in \calX_{\mathrm{in}}} w_i^{(1)}}{\min_{\bx_i\in \calX_{\mathrm{in}}} w_i^{(1)}} \prod_{k=1}^{K-1} (1+r_k C_0)
\]remains bounded for all $K$. 

The above inequality implies $\liminf_{k\to\infty}\min_{\bx_i\in \calX_{\mathrm{in}}} w_i^{(k)} > 0$.  Indeed, if there exists a subsequence with indices $\calK$  such that $\lim_{k\in\calK, k\to\infty}\min_{\bx_i\in \calX_{\mathrm{in}}} w_i^{(k)} = 0$, then it would follow from the above inequality and from $\lim_{k\rightarrow\infty}r_k=0$ that
\[
\lim_{\substack{k\to\infty \\ k\in\calK}}\max_{\bx_i\in \calX_{\mathrm{in}}} w_i^{(k)} = 0
\qquad\text{and}\qquad
\lim_{\substack{k\to\infty \\ k\in\calK}}\max_{\bx_i\in \calX_{\mathrm{out}}} w_i^{(k)} = \lim_{\substack{k\to\infty \\ k\in\calK}}(r_k\min_{\bx_i\in \calX_{\mathrm{in}}} w_i^{(k)})=0.
\]
Hence all weights would converge to zero, so   $\lim_{k\in\calK, k\to\infty}\bSigma^{(k)}=\lim_{k\in\calK, k\to\infty}\sum_{i=1}^N w_i^{(k)}\bx_i\bx_i^\top$ converges to zero, contradicting the normalization of the TME iterates such that $\tr(\bSigma^{(k)})=1$.

Finally, since $\lim_{k\rightarrow\infty}r_k=0$, the total weight assigned to the outliers becomes negligible for sufficiently large $k$, and $\bSigma^{(k)}$ is effectively a weighted sum of the inliers alone. Consequently,  for large $k$, the TME iteration for $\calX$ converges to the TME iteration for the set of inliers $\calX_{\mathrm{in}}$. By \cite[Theorem 1.2]{zhang2016robust}, the latter converges to the TME solution of $\calX_{\mathrm{in}}$, so the TME iteration for $\calX$ also converges to it.

\subsection{Proof of \Cref{lemma:largest_dssnr}}

We first introduce some notation. For any subspace \(L \neq \{0\} \), we denote
\[
f(\calX,L)
=
\frac{|\calX \cap L|}{\dim(L)}.
\]
Moreover, we denote
\[
P_{L^\perp}(\calX)
=
\left\{
P_{L^\perp}\bx
:
\bx \in \calX \setminus L
\right\}.
\]
Note that we removed all points in $L$ so that the projected set does not contain zeros.

We now construct a sequence of subspaces recursively. Let \(L_0\) be the trivial \(0\)-dimensional subspace. For each \(k \geq 0\), define
\[
V_k = \bigoplus_{i=0}^k L_i,
\]
and choose
\begin{equation}\label{eq:define_Lk}
L_{k+1}
=
\argmax_{L \, : \,  L\perp V_k, L\neq \{0\}}
f\bigl(P_{V_k^\perp}(\calX),L\bigr).
\end{equation}
If the maximizer is not unique, we choose among all maximizers one with the smallest dimension. We note that \(\dim(V_k)\) strictly increases with \(k\) and thus there exists \(K \in \N\) such that $V_K = \Sp(\calX)$.  
At this point, the iteration terminates.

For each \(k=1,\dots,K\), define
\[
\calX_k
:=
\left\{
\bx_i \in \calX :
\bx_i \in V_k \setminus V_{k-1}
\right\}.
\]
Since \(V_K=\mathbb R^D\), the dataset can be decomposed as $\calX
=
\calX_1 \cup \calX_2 \cup \cdots \cup \calX_K .$

We shall use the following properties, which we prove later in \Cref{sec:verify_3_prop}. 

\begin{enumerate}
    \item The first selected subspace satisfies
    $L_1=L_*$, 
    and \(\calX_1\) is the set of inliers.

    \item The sequence of ratios
    \[
    \frac{|\calX_k|}{\dim(L_k)}
    \]
    is nonincreasing in \(k\). Moreover, the inequality is strict between the first and second terms; that is,
    \[
    \frac{|\calX_K|}{\dim(L_K)}
    \leq
    \frac{|\calX_{K-1}|}{\dim(L_{K-1})}
    \leq
    \cdots
    \leq
    \frac{|\calX_2|}{\dim(L_2)}
    <
    \frac{|\calX_1|}{\dim(L_1)} .
    \]

    \item For every \(k=1,\dots,K\), the dataset 
    $\widehat{\calX}_k
    :=
    \left\{
    \bU_{L_k}^{\top}\bx
    :
    \bx \in \calX_k
    \right\} \subseteq \R^{d_k}$, where $d_k=\dim(L_k)$, 
    admits a unique TME solution, denoted by \(\bSigma_k \in S_{++}(d_k)\).
\end{enumerate}

Assuming that the three properties above hold, we construct a linear operator
\(\bT \colon \R^D \to \R^D\) that is block diagonal with respect to the decomposition
$\Sp(\calX) = L_1 \oplus L_2 \oplus \cdots \oplus L_K$. It depends on an arbitrary parameter $\epsilon>0$, which we will later tune. 
Specifically, for each \(k=1,\dots,K\), define the block of \(\bT\) on \(L_k\) by
\[
[\bT]_{L_k,L_k}
=
\epsilon^{1-k}\bSigma_k^{-1/2},
\]
where property 3 above provides $\bSigma_k$. 
That is, $[\bT\bx]_{L_k}=\epsilon^{1-k}\bSigma_k^{-1/2}[\bx]_{L_k}$, where $[\bx]_{L_k}=\bU_{L_k}^\top \bx\in \R^{d_k}$.

For any \(\bx\in \calX_k\), $[\bx]_{L_k}\neq 0$ and $[\bx]_{L_j}=0$ for all $j>k$.  
For a small $\epsilon>0$, the vector \(\bT\bx\) is dominated by its component in \(L_k\), i.e., $\epsilon^{1-k}\bSigma_k^{-1/2}[\bx]_{L_k}$. Indeed, the higher components are zeros and the next lower one is $\epsilon \cdot \epsilon^{1-k}\bSigma_{k-1}^{-1/2}[\bx]_{L_{k-1}}$. We thus expect an error of order \(O(\epsilon)\) and we show it rigorously as follows: 
For $\bx\in\calX_k$,
\[
\bx=\sum_{j=1}^k\bU_{L_j}[\bx]_{L_j}, \bT\bx=\sum_{j=1}^k\bU_{L_j}\epsilon^{1-j}\bSigma_j^{-1/2}[\bx]_{L_j},
\]
so $\bT\bx=\by^{(1)}+\by^{(2)}$, where $\by^{(1)}=\bU_{L_k}\epsilon^{1-k}\bSigma_k^{-1/2}[\bx]_{L_k}$ and $\by^{(2)}=\sum_{j=1}^{k-1}\bU_{L_j}\epsilon^{1-j}\bSigma_j^{-1/2}[\bx]_{L_j}$. We note that  $\by^{(1)}\perp \by^{(2)}$, and for a sufficiently small $\epsilon$, $\|\by^{(1)}\|=O(\epsilon^{1-k})$ and $\|\by^{(2)}\|\leq O(\epsilon^{2-k})$, which implies $\|\by^{(2)}\|/\|\by^{(1)}\|\leq O(\epsilon)$. Then, 
\begin{align*}
&\left\|\frac{\bT\bx}{\|\bT\bx\|}-\frac{\by^{(1)}}{\|\by^{(1)}\|}\right\|
=\left\|\frac{\by^{(2)}}{\|\bT\bx\|}+\frac{\by^{(1)}}{\|\bT\bx\|}-\frac{\by^{(1)}}{\|\by^{(1)}\|}\right\|\leq\left\|\frac{\by^{(2)}}{\|\bT\bx\|}\right\|+\left\|\frac{\by^{(1)}}{\|\bT\bx\|}-\frac{\by^{(1)}}{\|\by^{(1)}\|}\right\| 
\\\leq&\left\|\frac{\by^{(2)}}{\|\by^{(1)}\|-\|\by^{(2)}\|}\right\|+\left\|\frac{\by^{(1)}}{\|\by^{(1)}\|-\|\by^{(2)}\|}-\frac{\by^{(1)}}{\|\by^{(1)}\|}\right\|\\
=&\frac{\|\by^{(2)}\|}{\|\by^{(1)}\|-\|\by^{(2)}\|}+\left|\frac{\|\by^{(1)}\|}{\|\by^{(1)}\|-\|\by^{(2)}\|}-\frac{\|\by^{(1)}\|}{\|\by^{(1)}\|}\right|=2\frac{\|\by^{(2)}\|}{\|\by^{(1)}\|-\|\by^{(2)}\|}\leq O(\epsilon). 
\end{align*}
In addition, when \(k=1\), this error term vanishes. Indeed, if $\bx\in\calX_1$, then $\bx=\bU_{L_1}[\bx]_{L_1}$ and $\by^{(2)}=0$. Hence,
\[
\frac{\bT\bx}{\|\bT\bx\|}=\frac{\by^{(1)}}{\|\by^{(1)}\|} +\mathbf 1_{\{k\geq 2\}}\,O(\epsilon) =\bU_{L_k}\frac{\bSigma_k^{-1/2}[\bx]_{L_k}}{\|\bSigma_k^{-1/2}[\bx]_{L_k}\|}+\mathbf 1_{\{k\geq 2\}}\,O(\epsilon),
\]
and
\[
\frac{(\bT\bx)(\bT\bx)^\top}
{\|\bT\bx\|^2}
=
\bU_{L_k}
\frac{
\bSigma_k^{-1/2}
[\bx]_{L_k}
[\bx]_{L_k}^{\top}
\bSigma_k^{-1/2}
}{
[\bx]_{L_k}^{\top}
\bSigma_k^{-1}
[\bx]_{L_k}
}
\bU_{L_k}^{\top}
+
\mathbf 1_{\{k\geq 2\}}\,O(\epsilon).
\]
Recalling that \(\bSigma_k\) is the unique TME solution for $\widehat{\calX}_k$ (see property 3 above) and combining this property with \eqref{eq:fixed_from_Ftilde}, yields
\[
\sum_{\bz\in \widehat{\calX}_k}
\frac{\bSigma_k^{-1/2} \bz\bz^\top\bSigma_k^{-1/2}}
{\bz^\top\bSigma_k^{-1}\bz}
=
\frac{|\calX_k|}{\dim(L_k)}\bI_{d_k\times d_k}.
\]
Noting that $\widehat{\calX}_k=\{[\bx]_{L_k}: \bx \in \calX_k\}$, we obtain
\begin{multline*}
\sum_{\bx\in\calX_k}
\frac{(\bT\bx)(\bT\bx)^\top}
{\|\bT\bx\|^2}
=\bU_{L_k}\left(\sum_{\bz\in \widehat{\calX}_k}
\frac{\bSigma_k^{-1/2} \bz\bz^\top\bSigma_k^{-1/2}}
{\bz^\top\bSigma_k^{-1}\bz}\right)\bU_{L_k}^\top+
\mathbf 1_{\{k\geq 2\}}\,O(\epsilon)\\=
\frac{|\calX_k|}{\dim(L_k)}
\bP_{L_k}
+
\mathbf 1_{\{k\geq 2\}}\,O(\epsilon).    
\end{multline*}
Assume first that \(k=1\) and recall that property 1 above implies \(L_1=L_*\). Combining this assumption with the above equation, applied with $k=1$, yields \eqref{eq:equality_in3}. 
In order to conclude \eqref{eq:equality_out3}, we sum the above equation over all \(2\leq k\leq K\) 
and define 
$$\dssnr_{\mathrm{out}}=\frac{D}{N}\left\|\sum_{k=2}^K \frac{|\calX_k|}{\dim(L_k)}\bP_{L_k}+O(\epsilon)\right\|.$$ 
We note that property 2 above implies that $\dssnr_{\mathrm{out}}=\frac{D}{N}\frac{|\calX_2|}{\dim(L_2)}+O(\epsilon)$. Finally, the strict inequality in  Property~2 implies that we can choose $\epsilon$ sufficiently small such that 
$\dssnr_{\mathrm{out}} < \dssnr_{\mathrm{in}}$, which indeed concludes \eqref{eq:equality_out3}.

\subsubsection{Verification of the three properties}
\label{sec:verify_3_prop}
The first property follows from the definition of $L_1$ (see \eqref{eq:define_Lk}), where $f\bigl(P_{V_0^\perp}(\calX),L\bigr)=\frac{|\calX \cap L|}{\dim(L)}$. By assumption \eqref{eq:assumption3_thm}, the corresponding maximizer must be
$L_1=L_*$. This observation and the fact that $\calX$ does not contain the origin imply that \(\calX_1\) is the set of inliers.

We next prove the second property by contradiction. Assume that for some \(k\geq 1\),
$
\frac{|\calX_{k+1}|}{\dim(L_{k+1})}
>
\frac{|\calX_k|}{\dim(L_k)}.
$
Note that by definition, $f\bigl(P_{V_k^\perp}(\calX),L\bigr)=\frac{|\{P_{V_k^\perp}\bx \,:\, \bx\in \calX\setminus V_k\}\cap L|}{\dim(L)}=\frac{|\calX\cap ((V_k\oplus L)\setminus V_k)|}{\dim(L)}$, where the last equality follows from observing that if $P_{V_k^\perp}\bx \in L$, then $\bx = P_{V_k}\bx +P_{V_k^\perp}\bx \in V_k \oplus L$. Similarly, replacing $V_k^\perp$ by $V_{k-1}^\perp$ and $L$ by $L_k\oplus L_{k+1}$, while noting that $V_{k+1}=V_{k-1} \oplus L_k\oplus L_{k+1}$, $L_{k} \perp L_{k+1}$ and $\calX_k \cap \calX_{k+1}=\emptyset$,
\begin{align}\nonumber
f\bigl(P_{V_{k-1}^{\perp}}(\calX),L_k\oplus L_{k+1}\bigr)
&=\frac{|(\calX\cap V_{k+1})\setminus V_{k-1}|}
{\dim(L_k\oplus L_{k+1})}=
\frac{|\calX_k|+|\calX_{k+1}|}
{\dim(L_k)+\dim(L_{k+1})}  \\
&>
\frac{|\calX_k|}{\dim(L_k)} =
f\bigl(P_{V_{k-1}^{\perp}}(\calX),L_k\bigr).\label{eq:verification_property2}
\end{align}
This contradicts the definition of \(L_k\) as a maximizer in \eqref{eq:define_Lk}. Therefore,
$\frac{|\calX_{k+1}|}{\dim(L_{k+1})}
\leq
\frac{|\calX_k|}{\dim(L_k)}$
 holds for every \(k\geq 1\).

It remains to show that the  inequality is strict for $k=1$. If, on the contrary, $
\frac{|\calX_2|}{\dim(L_2)}
=
\frac{|\calX_1|}{\dim(L_1)}$, 
then the same averaging argument as in \eqref{eq:verification_property2} gives
\[
f(P_{V_0^{\perp}}(\calX),L_1\oplus L_2)
=
f(P_{V_0^{\perp}}(\calX),L_1),
\]
that is,
\[
\frac{|\calX \cap L_*|}{\dim(L_*)}=\frac{|\calX \cap (L_2\oplus L_*)|}{\dim(L_2\oplus L_*)}
\]
contradicting the strict inequality in \eqref{eq:assumption3_thm} when $L=L_1\oplus L_2$. This concludes the proof of property 2. 

Finally, we prove property 3. By the definition of \(L_{k+1}\), the set
$
\widehat{\calX}_{k+1}
=
\left\{
\bU_{L_{k+1}}^\top \bx:
\bx\in \calX_{k+1}
\right\}, 
$
which is contained in $\R^{d_{k+1}}$,
has the following property: 
\begin{equation}\label{eq:third_property_proof}
\text{for every proper nontrivial subspace \(L\subset \R^{d_{k+1}}\),} \,\, |\widehat{\calX}_{k+1}\cap L|
<
\frac{|\widehat{\calX}_{k+1}|}{\dim(L_{k+1})}\dim(L).
\end{equation}
Indeed, suppose  \( \b0 \subsetneq L \subsetneq \R^{d_{k+1}}\) contains at least 
${|\widehat{\calX}_{k+1}|\dim(L)}/{\dim(L_{k+1})}$
points. 
Let $\hat{L}\in\R^D$ be the subspace spanned by $\{\bU_{L_{k+1}}\by: \by\in L\}$. Note that $\{\bU_{L_{k+1}}\by: \by\in L\} \subset V_{k+1} \setminus V_k$ and $\dim(\hat{L})=\dim(L)<\dim(L_{k+1})$. Consequently,   
\[
f\bigl(P_{V_k^\perp}(\calX),\hat{L}\bigr)
\geq \frac{\frac{|\widehat{\calX}_{k+1}|}{\dim(L_{k+1})}\dim(L)}{\dim(L)}=\frac{|{\calX}_{k+1}|}{\dim(L_{k+1})}=
f\bigl(P_{V_k^\perp}(\calX),L_{k+1}\bigr).
\]
On the other hand, note that \eqref{eq:define_Lk} implies the reverse inequality and thus equality in the equation above. Moreover, the tie-breaking rule defining $L_{k+1}$ requires in this case that $\dim(L_{k+1}) \leq \dim(\hat{L})$, which results in contradiction. We thus conclude \eqref{eq:third_property_proof}. Combining this with \cite[Theorem~1.2]{zhang2016robust} we establish the existence and uniqueness of the TME estimator for each set \(\widehat{\calX}_{k+1}\). 

\section{Discussion}
\label{sec:conclusion}

In this paper, we presented a rigorous geometric analysis of the stability and convergence properties of TME for the RSR problem. By extending the foundational stability conditions of \cite{zhang2016robust}, we established a generalized framework that guarantees the exact recovery of the underlying inlier subspace under a stability condition that is less restrictive than standard general position assumptions. 

Most notably, our analysis resolves an open problem in the literature by characterizing the algorithmic behavior of TME at the critical boundary $\dssnr = 1$, thereby formally establishing a sharp computational phase transition for an RSR algorithm. 

From a technical standpoint, these guarantees were established by formulating an equivalent Majorization-Minimization (MM) framework for the TME iterates and introducing a novel decomposition of the objective function, which allowed us to precisely isolate the asymptotic behavior of the inliers from the outliers. 

We note that this work has direct implications for the setting of noiseless hybrid linear modeling (or subspace clustering) with outliers. In this setting, data points lie on a union of multiple subspaces, corrupted by additional outliers that do not belong to any of these subspaces. An early theoretical approach to this problem~\cite{lp_recovery_part1_11} focused on identifying a ``most significant'' subspace. This naturally motivates an iterative procedure: finding the dominant subspace, removing its inliers, and repeating the process for the remaining data. Our current work rigorously clarifies the weakest generic conditions under which such an iterative process can succeed. Specifically, \cref{item:dssnr1,item:convergence} of \Cref{thm:zhang16_extend} must hold at each step of the iteration; if they fail, one cannot expect to efficiently recover the next most significant subspace. Furthermore, our work implies that iterative application of TME exactly recovers these multiple subspaces when these weak generic conditions hold. 

For future work, an important direction is the formulation of a fully noisy RSR setting to study its computational limits and approximate recovery properties. Transitioning from exact geometric recovery to this noisy regime fundamentally shifts the problem to average-case statistical inference. While recent advances in theoretical computer science have begun to map out statistical-computational gaps for robust subspace recovery (e.g., via Sum-of-Squares lower bounds~\cite{Bakshi2021}), these frameworks predominantly operate under a strictly noiseless signal model where inliers lie perfectly on the target subspace, or allow a very small perturbation by noise. Extending these computational hardness thresholds to the true noisy regime, and bridging them with efficient algorithms, remains a compelling, wide-open frontier for future research.

A more direct open question for future research is the formal convergence rate of TME when $\dssnr=1$ and \cref{item:convergence} of \cref{thm:zhang16_extend} holds. Note that our proof for the case $\dssnr > 1$ establishes linear convergence. Indeed, this follows directly from \eqref{eq:ratio_convergence}. This linear convergence cannot be implied by \cite{franks2020rigorous}, since their result guarantees linear convergence only to a strictly nonsingular matrix, whereas in our case, the global limiting matrix is singular. Numerical experiments seem to indicate sublinear convergence when $\dssnr=1$. It will be interesting to mathematically quantify this sublinear rate. This expected phenomenon of marked deceleration at the critical boundary mirrors the ``critical slowing down'' commonly observed in dynamical systems at a bifurcation boundary \cite{strogatz2014nonlinear}. In the context of iterative optimization, this corresponds to the spectral radius of the algorithmic map approaching unity at the boundary of stability. A mathematically analogous phenomenon regarding the degradation of convergence rates for iterative methods appears in \cite{ortega2000iterative}.


\section*{Acknowledgments}
GL was  supported by NSF award DMS-2152766  and TZ was supported by NSF award DMS-2318926.

\appendix
\section{Supplemental Results}
\label{sec:supplement}
\subsection{Clarifications and Properties of the Stability Condition}
\label{sec:clarify_zhang16_extend}

In this section, we provide essential clarifications regarding the assumptions of our main result (\Cref{thm:zhang16_extend}) and the constraints present in prior works. We first establish that \cref{item:convergence} guarantees the existence and uniqueness of the target TME estimator, as well as necessary dimensional spanning properties.

\begin{lemma}[Existence and Uniqueness of the TME Solution for Inliers]
\label{lem:existence_uniqueness}
Under \cref{item:convergence} of \Cref{thm:zhang16_extend}, the dimension-reduced inlier dataset $\tilde{\calX}_{\mathrm{in}} := \{\bU_{L_*}^\top\bx: \bx\in\calX_{\mathrm{in}}\}$ admits a unique TME estimator $\bSigma_{\mathrm{in},*} \in S_{++}(d)$.
\end{lemma}
\begin{proof}
For any subspace $  \b0 \subsetneq L \subsetneq \R^d$, we evaluate the fraction of dimension-reduced inliers it contains:
\begin{equation}
\label{eq:teng_exist_cond}
\frac{|\tilde{\calX}_{\mathrm{in}} \cap L|}{\dim(L)} = \frac{|{\calX}_{\mathrm{in}} \cap \{\bU_{L_*}\bx: \bx\in L\}|}{\dim(\{\bU_{L_*}\bx: \bx\in L\})} = \frac{|{\calX} \cap \{\bU_{L_*}\bx: \bx\in L\}|}{\dim(\{\bU_{L_*}\bx: \bx\in L\})}.
\end{equation}
The first equality holds because the mapping $\bU_{L_*}$ preserves both cardinality and dimension. The second equality follows from the fact that outliers, by definition, do not lie in $L_*$; therefore, they cannot intersect the subspace $\{\bU_{L_*}\bx: \bx\in L\} \subset L_*$. 

We now apply \cref{item:convergence} of \Cref{thm:zhang16_extend} to the subspace $\{\bU_{L_*}\bx: \bx\in L\} \subsetneq \R^D$, yielding:
\begin{equation*}
\frac{|{\calX} \cap \{\bU_{L_*}\bx: \bx\in L\}|}{\dim(\{\bU_{L_*}\bx: \bx\in L\})} < \frac{|{\calX} \cap L_*|}{\dim(L_*)} = \frac{|\tilde{\calX}_{\mathrm{in}}|}{d}.
\end{equation*}
Combining this with \eqref{eq:teng_exist_cond} gives $\frac{|\tilde{\calX}_{\mathrm{in}} \cap L|}{\dim(L)} < \frac{|\tilde{\calX}_{\mathrm{in}}|}{d}$. According to Theorem 1.2 of \cite{zhang2016robust} (which builds upon \cite[Theorem 2]{Kent_Tyler88} and \cite[Proposition 1(a)]{tyler2005}), this strict inequality holding for all $\b0 \subsetneq L \subsetneq \R^d$ is the necessary and sufficient condition for the existence and uniqueness of $\bSigma_{\mathrm{in},*}$.
\end{proof}

\begin{lemma}[Implicit Dimensionality Constraints]
\label{lem:implicit_dimension}
If a noiseless inlier-outlier dataset $\calX$ satisfies \cref{item:dssnr1,item:convergence} of \Cref{thm:zhang16_extend}, then the total number of points satisfies $N > D$. Furthermore, if $\dssnr=1$, the dataset fully spans the ambient space, i.e., $\operatorname{span}(\calX) = \R^D$.
\end{lemma}
\begin{proof}
Consider any 1-dimensional subspace (a line) $L$ that contains at least one point from $\calX$. For this line, $\frac{|\mathcal{X} \cap L|}{\dim(L)} \geq 1$. By \cref{item:convergence}, this must be strictly less than $\frac{|\mathcal{X} \cap L_*|}{\dim(L_*)} = \frac{n_1}{d}$. Therefore, we must have $n_1 > d$. Since \cref{item:dssnr1} requires $\dssnr \geq 1$, it follows from the definition of $\dssnr$ that $n_0 \geq \frac{n_1(D-d)}{d} > D-d$. Consequently, the total number of points is $N = n_0 + n_1 > D$.

Next, consider the case where $\dssnr=1$. Let $L = \operatorname{span}(\calX)$ and assume for the sake of contradiction that $L \subsetneq \R^D$. Note that for any value of $\dssnr$, straightforward algebraic manipulation yields:
\begin{equation}\label{eq:dssnr_N}
\frac{n_1/d}{N/D} = \frac{n_1/d}{(n_1+n_0)/D} = \frac{n_1 D}{\left(n_1 + n_1\frac{D-d}{d\cdot\dssnr}\right)d} = \frac{D}{(D-d)/\dssnr + d}.
\end{equation}
When $\dssnr=1$, \eqref{eq:dssnr_N} simplifies to $\frac{n_1/d}{N/D} = 1$, implying $\frac{N}{D} = \frac{|\calX\cap L_*|}{\dim(L_*)}$. However, if $L \subsetneq \R^D$, then:
\begin{equation*}
\frac{|\calX\cap L|}{\dim(L)} = \frac{N}{\dim(L)} > \frac{N}{D} = \frac{|\calX\cap L_*|}{\dim(L_*)},
\end{equation*}
which directly contradicts \cref{item:convergence}. Thus, we must have $\operatorname{span}(\calX) = \R^D$.
\end{proof}

Finally, we clarify why the theoretical guarantees for the deterministic algorithm in \cite{hardt2013algorithms} strictly require the conditions outlined in \Cref{thm:hm_doable}.

\begin{remark}[Necessity of $n_0 \geq D-d$ and $N>D$ for \cite{hardt2013algorithms}]
\label{rem:hardt_moitra_conditions}

We start with the necessity of having at least $D-d$ outliers. We first claim that without this requirement Condition 7 in \cite{hardt2013algorithms} imposes no effective restriction on the outliers. We recall that this condition states that a subset of $D$ points is linearly independent if and only if at most $d$ of them are inliers. 
Thus, if there are fewer than $D-d$ outliers, any subset of $D$ points must contain at least $d+1$ inliers. Since these $d+1$ inliers lie on a $d$-subspace they must be dependent. Consequently, Condition 7 is satisfied trivially and imposes no effective restriction on the outliers.

Next, we note that if two outliers are dependent, then Claim 10 in \cite{hardt2013algorithms} may fail. Claim 10 can be translated into our notation as follows: suppose that
\begin{equation}\label{eq:claim10}
\bx_{i_1}c_{i_1}+\cdots+\bx_{i_D}c_{i_D}=0
\end{equation}
for $D$ data points $\bx_{i_1},\ldots,\bx_{i_D}$, where $1 \leq i_1 < i_2< \ldots <i_D \leq N$, and scalars
$c_{i_1},\ldots,c_{i_D}$ that are not all zero, and let $\calI=\{i_j: 1\leq j\leq D, c_{i_j}\neq 0\}$ denote the set of
indices corresponding to the nonzero coefficients. Then $|\calI|\ge d+1$, and
moreover, for every $k\in\calI$, the point $\bx_k$ is an inlier.

Now, consider the example where $\bx_1$ and $\bx_2$ are two outliers satisfying
$\bx_1=-\bx_2$. Then \eqref{eq:claim10} holds with
$i_j=j$ for all $1\leq j\leq D$, and
$(c_1,\cdots,c_D)=(1,1,0,\ldots,0).$
In this case, we have $\calI=\{1,2\}$. Clearly, it violates Claim 10, which implies that for any $i \in \calI$, $\bx_i$ is an outlier. We note that the proof of Theorem \ref{thm:hm_doable} in \cite{hardt2013algorithms} (which is referred to there as Theorem~22) relies on Claim 10. Therefore, without the assumption $n_0 \geq D-d$ the proof of this theorem generally fails. 

Lastly, we establish the necessity of $N>D$ for Theorem \ref{thm:hm_doable}, or equivalently, 
\cite[Theorem 22]{hardt2013algorithms}. We note that the formulation of \cite[Theorem 22]{hardt2013algorithms} already requires $N\ge D$. Moreover,
when $N=D$, the assumption $\dssnr>1$ of \cite[Theorem 22]{hardt2013algorithms} implies that the number
of outliers must be strictly smaller than $D-d$. This contradicts the earlier
discussion that at least $D-d$ outliers are necessary.

\end{remark}

\subsection{Proof of \Cref{prop:moitra}}
The proof uses the fact that the number of inliers in $\calX$ is at least $d+1$, which follows from $N > D$ and $\dssnr \geq 1$. We will first show that 
\begin{equation}
\label{eq:basic_fact_L_subset_Lstart}
\text{any subspace } \ L \subsetneq L_* \ \text{ contains at most } \ \dim(L) \ \text{ inliers.}    
\end{equation}
Suppose, on the contrary, that $L$ contains at least $\dim(L)+1$ inliers. Since these points lie in a subspace of dimension $\dim(L)$, they must be linearly dependent. By combining this set with an additional $d - (\dim(L)+1)$ inliers and $D - d$ outliers, we obtain a set of $D$ points that is linearly dependent while containing $d$ inliers, contradicting 
the assumption that $\calX$ is in general position with respect to $L_*$ (every set of $D$ points is linearly independent if and only if it contains at most $d$ inliers).

We now prove the proposition by showing that, for any subspace $L\neq L_*$, \eqref{eq:assumption3_thm} holds, 
by considering two cases.

\begin{itemize}
    \item $L$ is a subspace that does not contain $L_*$:     
    We will first show that $L$ contains at most $\dim(L)$ data points. 
    Suppose instead that a set $\calX_0$ of $\dim(L)+1$ data points lies in $L$. According to \eqref{eq:basic_fact_L_subset_Lstart}, at most $\dim(L \cap L_*)$ of these points can be inliers. Let $x$ denote the number of inliers in $\calX_0$, then $x\leq \dim(L \cap L_*)<d$  and the number of outliers in $\calX_0$ is $\dim(L)+1-x$. 
    We construct $\calX_1$, a set of $D$ points that contradicts 
    the assumption that $\calX$ is in general position with respect to $L_*$ (every set of $D$ points is linearly independent if and only if it contains at most $d$ inliers) 
    in two complementary cases: 
    \begin{itemize}
    \item When $D-d\geq \dim(L)+1-x$, $\calX_1$ contains $\calX_0$ and an additional $d-x$ inliers from $\calX$ (which is possible as there are at least $d$ inliers in $\calX$), and an additional $D-d-(\dim(L)+1-x)$ outliers from $\calX$ (which is possible as there are at least $D-d$ outliers in $\calX$). For this case,  $\calX_1$ contains $d$ inliers. 
    \item When $D-d< \dim(L)+1-x$,  $\calX_1$ contains $\calX_0$ and an additional $D-(\dim(L)+1)$ inliers. We can select the additional $D-(\dim(L)+1)$ inliers because the number of inliers that are in $\calX\setminus \calX_0$ is at least $d-x$, which is larger than $D-(\dim(L)+1)$ as $x+(D-(\dim(L)+1))=d+\left((D-d)-(\dim(L)+1-x)\right)<d$.
    For this case, the number of inliers in  $\calX_1$ is $x+(D-(\dim(L)+1))$, which is smaller than $d$ by the argument above.
    \end{itemize}
    
    Then  $\calX_1$ is a set of $D$ points containing at most $d$ inliers, but also contains $\calX_0$, a linearly dependent subset, and hence  $\calX_1$ is also linearly dependent. This  contradicts the assumption that $\calX$ is in general position with respect to $L_*$. Indeed, every set of $D$ points is linearly independent if and only if it contains at most $d$ inliers. Consequently, we can assume in this case that $L$ contains at most $\dim(L)$ data points. 
    Therefore,
    \begin{equation}\label{eq:XL=1}
\frac{|\mathcal{X} \cap L|}{\dim(L)} \leq      \frac{\dim(L)}{\dim(L)} = 1.
    \end{equation}
    In comparison, applying \eqref{eq:dssnr_N}, $N>D$, and $\dssnr\geq 1$, \begin{equation}\label{eq:XL*>1}\frac{|\mathcal{X} \cap L_*|}{\dim(L_*)}>1.\end{equation} Combining \eqref{eq:XL=1} and \eqref{eq:XL*>1}, \eqref{eq:assumption3_thm} is proved.
    
    \item $L$ contains $ L_*$ and $L\neq L_*$, i.e., $L\supsetneq L_*$: We will first show that  $L$ contains at most $\dim(L)-d$ outliers. 
    Suppose instead that $\calX_0$ is a set of  $\dim(L)-d+1$ outliers that lie in $L$. Let  $\calX_1$ contain  $\calX_0$, any additional
    $ D-d-(\dim(L)-d+1)=D-\dim(L)-1$ 
    outliers from $\calX$ (it is possible as there are at least $D-d$ outliers from $\calX$), and any $d$ inliers. Then  $\calX_1$ is a set of $D$ points containing exactly $d$ inliers. However, $L\cap \calX_1$ contains $\calX_0$ and $d$ inliers, so at least
    $
    (\dim(L)-d+1)+d=\dim(L)+1
    $
 points of  $\calX_1$ lie in $L$, so the set $\calX_1$ is linearly dependent, again contradicting the assumption that $\calX$ is in general position with respect to $L_*$ (every set of $D$ points is linearly independent if and only if it contains at most $d$ inliers).
    
Then we conclude that, $L$ contains at most $\dim(L)-d$ outliers, that is, $|\calX\cap (L\setminus L_*)|\leq \dim(L)-d$. Then  \eqref{eq:assumption3_thm} can be verified: 
\[
\frac{|\mathcal{X} \cap L|}{\dim(L)}
=\frac{|\mathcal{X} \cap L_*|+|\calX\cap (L\setminus L_*)|}{\dim(L_*)+(\dim(L)-d)}\leq 
\frac{|\mathcal{X} \cap L_*|+(\dim(L)-d)}{\dim(L_*)+(\dim(L)-d)}
<
\frac{|\mathcal{X} \cap L_*|}{\dim(L_*)},
\]
where the last step follows from $\dim(L)-d>0$ and \eqref{eq:XL*>1}. \qed
\end{itemize}

\subsection{Proof of \Cref{prop:thirdassumption}}
For any subspace \(L\), we may decompose it as
\[
L = L_1 \oplus L_2, \quad 
L_1 = L \cap L_*, \quad
L_2 = L \cap L_1^\perp .
\]
By the general position assumption on the inliers, if \(L_1 \neq L_*\), the number of inliers contained in \(L\) is at most \(\dim(L_1)\). Furthermore, if \(L_1 = L_*\), then \(L\) contains exactly \(n_1\) inliers.

Similarly, by the general position assumption on the projected outliers and the fact that for any \(\bx \in L\),
\[
P_{L_*^\perp}\bx \in P_{L_*^\perp}(L) = P_{L_*^\perp}(L_2),
\] if \(P_{L_*^\perp} L_2 \neq L_*^\perp\), then the number of outliers contained in \(L\) is at most \(\dim(L_2)\). Furthermore, if \(P_{L_*^\perp} L_2 = L_*^\perp\), then \(L\) contains exactly \(n_0\) outliers.

We now prove \Cref{prop:thirdassumption} by considering four cases.

\paragraph{Case 1: \(L_1 \neq L_*\) and \(P_{L_*^\perp}L_2 \neq L_*^\perp\).}  
Note that \(L\) contains at most \(\dim(L_1)+\dim(L_2)=\dim(L)\) points.  
Since \(N > D\), the fraction of points in \(L\) is at most
\[
\frac{\#(L \cap \calX)}{N}\leq \frac{\dim(L)}{N} < \frac{\dim(L)}{D}.
\]

\paragraph{Case 2: \(L_1 = L_*\) and \(P_{L_*^\perp}L_2 \neq L_*^\perp\).}  
In this case \(L\) contains at most \(n_1+\dim(L_2)\) points.  
We further note that since \(\dssnr=1\), \(\frac{n_1}{N}=\frac{d}{D}\).  Lastly, observe that 
\(\dim(L)=\dim(L_1)+\dim(L_2)=d+\dim(L_2)\). Combining these observations and the inequality $N>D$ results in
\[
\frac{\#(L \cap \calX)}{N} \leq
\frac{n_1+\dim(L_2)}{N} =\frac{d}{D}+\frac{\dim(L_2)}{N} < \frac{d+\dim(L_2)}{D}=\frac{\dim(L)}{D}.
\]

\paragraph{Case 3:  \(L_1 \neq L_*\) and \(P_{L_*^\perp}L_2 = L_*^\perp\).}  
In this case \(L\) contains at most \(\dim(L_1)+n_0\) points.  Note that since \(\dssnr=1\), \(\frac{n_0}{N}=\frac{D-d}{D}\). 
Lastly, since \(\dim(L_2)=D-d\),
\[
\dim(L)=\dim(L_1)+\dim(L_2)=\dim(L_1)+(D-d).
\]
Combining these estimates and $N>D$ yields 
\[
\frac{\#(L \cap \calX)}{N} \leq \frac{\dim(L_1)+n_0}{N}
   = \frac{\dim(L_1)}{N} + \frac{D-d}{D}
   < \frac{\dim(L_1)}{D} + \frac{D-d}{D}
   = \frac{\dim(L)}{D}.
\]

\paragraph{Case 4: \(L_1 = L_*\) and \(P_{L_*^\perp}L_2 = L_*^\perp\).}   
In this case $\dim(L)=\dim(L_1)+\dim(L_2)=d+(D-d)=D$, and thus $L=\R^D$, which contradicts the assumption that $L$ is a proper subspace. 
Combining all four cases completes the proof of the lemma. \qed

\subsection{Existence and Uniqueness of the TME Solution for Outliers}
\label{sec:existence_uniqueness_outliers}
We recall that \eqref{eq:x_in_out_def} defines 
$\widehat{\calX}_{\mathrm{out}}=\{\bU_{L_*^\perp}^\top\bx: \bx\in\calX_{\mathrm{out}}\}$  
and we remind the reader of the definition of $\tilde{\calX}_{\mathrm{in}}$. The existence and uniqueness of the TME estimator for $\tilde{\calX}_{\mathrm{in}}$ was proved in \cref{lem:existence_uniqueness}. Here we verify the existence and uniqueness of $\widehat{\calX}_{\mathrm{out}}$.
\begin{lemma}
\label{lem:existence_uniqueness_outliers}
Under \cref{item:convergence} of \Cref{thm:zhang16_extend} and the condition $\dssnr=1$, the dimension-reduced outlier dataset $\widehat{\calX}_{\mathrm{out}} := \{\bU_{L_*^\perp}^\top\bx: \bx\in\calX_{\mathrm{out}}\}$ admits a unique TME estimator $\bSigma_{\mathrm{out},*} \in S_{++}(D-d)$.
\end{lemma}

\begin{proof}
According to \cite[Theorem 1.2]{zhang2016robust}, the unique estimator exists if and only if for all proper subspaces $ \b0 \subsetneq L' \subsetneq \R^{D-d}$, the fraction of points satisfies:
$$ \frac{|\widehat{\calX}_{\mathrm{out}} \cap L'|}{\dim(L')} < \frac{|\widehat{\calX}_{\mathrm{out}}|}{D-d} = \frac{n_0}{D-d}. $$

We proceed by contradiction. Assume there exists a proper subspace $L' \subsetneq \R^{D-d}$ such that $|\widehat{\calX}_{\mathrm{out}} \cap L'| \geq n_0 \frac{\dim(L')}{D-d}$. We construct a lifted subspace in $\R^D$ defined by:
$$ L_{\mathrm{lift}} = \left\{ \bx + \bU_{L_*^\perp}\by : \bx \in L_*, \ \by \in L' \right\}. $$
Since $L_*$ and the range of $\bU_{L_*^\perp}$ are orthogonal complements, the dimension of this new subspace is $\dim(L_{\mathrm{lift}}) = d + \dim(L')$. Furthermore, because $L' \subsetneq \R^{D-d}$, $L_{\mathrm{lift}}$ is strictly a proper subspace of $\R^D$.

Next, we count the total number of points from $\calX$ contained within $L_{\mathrm{lift}}$. By definition, it contains all $n_1$ inliers because $L_* \subset L_{\mathrm{lift}}$. It also contains all outliers $\bx \in \calX_{\mathrm{out}}$ whose projections $\bU_{L_*^\perp}^\top\bx$ lie in $L'$. By our contradiction hypothesis, there are at least $n_0 \frac{\dim(L')}{D-d}$ such points. Thus, the total number of points in $L_{\mathrm{lift}}$ is at least:
$$ |\calX \cap L_{\mathrm{lift}}| \geq n_1 + n_0\frac{\dim(L')}{D-d}. $$

Recall that the condition $\dssnr=1$ implies $\frac{n_1}{d} = \frac{N}{D}$. Since $N = n_1 + n_0$ and $D = d + (D-d)$, it algebraically follows that $\frac{n_0}{D-d} = \frac{N}{D} = \frac{n_1}{d}$. Substituting this into our inequality yields:
$$ |\calX \cap L_{\mathrm{lift}}| \geq \frac{n_1}{d}d + \frac{n_1}{d}\dim(L') = \frac{n_1}{d}\bigl(d+\dim(L')\bigr) = \frac{n_1}{d}\dim(L_{\mathrm{lift}}). $$

Rearranging this gives:
$$ \frac{|\calX \cap L_{\mathrm{lift}}|}{\dim(L_{\mathrm{lift}})} \geq \frac{n_1}{d}. $$
However, \cref{item:convergence} of \Cref{thm:zhang16_extend} strictly requires that for any proper subspace of $\R^D$, the fraction of points is strictly less than $\frac{n_1}{d}$. We have thus reached a contradiction. Therefore, no such subspace $L'$ can exist, and the existence and uniqueness of $\bSigma_{\mathrm{out},*}$ is guaranteed.
\end{proof}

\bibliographystyle{plain}
\bibliography{RSR-ref}

\end{document}